\RequirePackage{fix-cm}
\documentclass[smallextended]{svjour3}       % onecolumn (second format)
\smartqed  % flush right qed marks, e.g. at end of proof
\usepackage{graphicx}
\usepackage{tabularx}
\usepackage{color}
\usepackage{mathtools,amsfonts,amssymb,amsmath,bm}
\usepackage[normalem]{ulem}

\renewcommand{\d}{{\rm d}}
\newcommand{\e}{{\rm e}}

\begin{document}

% \title{A theoretical model of Dengue fever transmission to analyse extrinsic and intrinsic incubation period via a system of delay differential equation
% }
\title{Dynamical behavior and bifurcation analysis for a theoretical model of dengue fever transmission with  incubation period and delayed recovery
}

\titlerunning{Dynamical behavior and bifurcation analysis for a theoretical model}
%\titlerunning{A theoretical model of Dengue fever transmission with delay}        % if too long for running head

\author{Burcu G\"{u}rb\"{u}z{$^{*, 1,2}$} \and Ayt\"{u}l G\"{o}k\c{c}e{$^{3}$} \and Segun I. Oke{$^{4}$} \and  Michael O. Adeniyi{$^{5}$} \and  Mayowa M. Ojo{$^{6}$}
}

\institute{
B. G\"{u}rb\"{u}z{$^*$} (Corresponding author)\at
{$^{1}$}Institute of Mathematics, Johannes Gutenberg-University Mainz, 55128 Mainz, Germany\\
{$^{2}$Institute for Quantitative and Computational Biosciences (IQCB), Johannes Gutenberg University Mainz, 55128 Mainz, Germany}\\
              \email{burcu.gurbuz@uni-mainz.de}
\and
A. G\"{o}k\c{c}e \at
{$^{3}$}Department of Mathematics, Ordu University, 52200 Ordu, Turkey.\\
              \email{aytulgokce@odu.edu.tr}           %  \\
\and
S. I. Oke \at
{$^{4}$}Department of Mathematics, Ohio University, Athens, 45701 OH, USA.\\
	       \email{segunoke2016@gmail.com}
\and
M. O. Adeniyi \at
{$^{5}$}Mathematical Sciences Department, Lagos State University of Science and Technology,  Ikorodu, Lagos, Nigeria.\\
	       \email{michaeladeniyi1976@gmail.com}
\and
M. M.Ojo \at
{$^{6}$}Department of Mathematical Sciences, University of South Africa, South Africa.\\
            \email{mmojomth@gmail.com}
}
\date{Received: date / Accepted: date}

\maketitle

\begin{abstract}
As offered by the World Health Organisation (WHO), close to half of the population in the world’s resides in dengue-risk zones. Dengue viruses are transmitted to individuals by Aedes mosquito species infected bite (Ae. Albopictus of Ae. aegypti). These mosquitoes can transmit other viruses, including Zika and Chikungunya. In this research, a mathematical model is formulated to reflect different time delays considered in both extrinsic and intrinsic incubation periods, as well as in the recovery periods of infectious individuals. Preliminary results for the non-delayed model  including positivity and boundedness of solutions, non-dimensionalization and equalibria analysis are presented. The threshold parameter (reproduction number) of the model is obtained via next generation matrix schemes. The stability analysis of the model  revealed that various dynamical behaviour can be observed depending on delay parameters, where in particular  the effect of delay in the recovery time of infectious individuals may lead to substantial changes in the dynamics. The ideas presented in this paper can be applied to other infectious diseases, providing qualitative evaluations for understanding time delays influencing the transmission dynamics.

\end{abstract}

\section{Introduction}
\label{Sec:Introduction}

Dengue fever has recently been identified as the most rapidly transmitted arboviral disease \cite{guzman2015dengue,ojo2017lyapunov}.
The disease has been identified as a major public health concern in tropical and subtropical countries, with a significant impact on population health, as evidenced by several authoritative studies and reports \cite{world2011comprehensive,world2011report,who2014selection}. Although, the dengue virus is common in the subtropical and tropical regions, it has been identified in approximately 128 countries, including those in North America and Europe, which are predominantly non-tropical countries.
Reports indicate that approximately 390 million individuals are infected with dengue annually, resulting in 500,000 hospitalizations and 20,000 deaths from complications of the disease \cite{who2014selection}, \cite{world2011report}. \\
\noindent The economic consequences of the fever epidemic have had a significant impact on several nations. From 2000 to 2007, the total annual cost in the United States, including both hospitalizations and ambulatory cases, was estimated at USD 2.1 billion \cite{shepard2011economic}.
Therefore, it is essential to ensure the implementation of health policy and the establishment of disease control schemes are needed \cite{luh2018economic}.\\
%% Dengue mechanism is the following:
\noindent The Aedes albopictus and Aedes aegypti are the two main species vectors  \cite{world2011comprehensive}. Four serotypes of the virus can cause dengue fever: DEN-1, DEN-2, DEN-3, and DEN-4.
Dengue fever is primarily transmitted by Aedes aegypti and Aedes albopictus mosquitoes, which become infected when they bite a person carrying the dengue virus. Following an incubation period of 8 to 12 days, the virus disseminates to the mosquito's salivary glands. This establishes the potential for transmission to other humans through subsequent bites. In humans, the virus incubates for 4 to 10 days before symptoms manifest.
The symptoms of the disease include back pain, high fever, skin rash, muscular pains, nausea, eyes soreness, vomiting, face redness, red eyes, severe malaise, extreme weakness, and death.
During this viremic phase, infected individuals can transmit the virus to other mosquitoes.
Environmental factors, such as temperature and humidity, significantly influence the transmission dynamics, which affect mosquito populations and breeding sites \cite{gubler1998dengue,guzman2002effect}. In light of the increasing incidence of dengue virus disease in recent years, numerous studies have been undertaken to develop an optimally efficacious vaccine. Meanwhile, a satisfactory preventive vaccine remains unavailable; however, a partially active vaccine is available \cite{east2016world,tan2018autochthonous}. Thus, scientists rely on management policies and control mechanisms to eliminate viral transmission.

\noindent Epidemiological formulations are a fundamental tool in gaining deeper insight into the dynamics of infectious disease and the implementation of effective
control techniques. The establishment of such mathematical formulations reveals an important area concerning the spreading and control \cite{xiao2013dynamics,ndairou2021fractional}. Mathematical formulations for the dynamical interactions between the vector and the host were first proposed by Lotka \cite{rosshistory} and Ross \cite{ross1911prevention}. They were the first to examine the dynamic of host-vector malaria. The proposed variations in the framework can be adopted for the study of dengue fever \cite{agusto2018optimal,esteva1998analysis,hamdan2019analysis}. Many scholars have previously adopted deterministic compartmental formulations to investigate the dengue disease dynamical propagation in specific areas or countries \cite{abidemi2020optimal,agusto2018optimal,esteva1998analysis,syafruddin2013lyapunov}. Recent mathematical models of dengue fever include different aspects of the disease, such as the effects of vector control, treatment and mass awareness on the transmission dynamics \cite{naaly24}, specifically the birth-pulse mosquito population models with sex structures have been established \cite{zhangli24}, models with linear and nonlinear infection rate have been taken into account \cite{muthu24}, individual immunological variability has also been considered \cite{anam24}.\\
\noindent While significant progress has been made in SIRS modeling, numerous open problems persist within the field. Recently, several authors have performed complex numerical experiments for various infectious disease models to study the sign of the eigenvalues obtained from the characteristic equation, evaluated at the endemic equilibrium \cite{kooi2013bifurcation,steindorf2024symmetry,steindorf2022modeling,stollenwerk2022seasonally,zhang2022bifurcations}. This has facilitated the analysis of the impact of various parameter values on the dynamics of the system and proved that stability of the endemic equilibrium may lead to very complicated dynamics. In this paper, we focus on investigating the bifurcations of a dengue virus transmission model that considers the extrinsic and intrinsic incubation period, as well as the time delay in the recovery period. We demonstrate that the model is susceptible to various kinds of bifurcations, including transcritical bifurcation, period-doubling bifurcation, and Hopf bifurcation, as parameters change. In this context, we show that the delay in the recovery, $\tau_r$, is particularly important for highly complex dynamics. Taking such terms into account in disease modeling may also prove instrumental in understanding different epidemic diseases, including the recent COVID-19 pandemic. For example, by studying period-doubling bifurcations, we can potentially identify parameter regions where complex dynamics, including deterministic chaos, may arise \cite{stollenwerk2022seasonally}.
\\
\noindent In this work, the authors investigate the impact of intrinsic and extrinsic incubation times on dengue fever infection dynamics by using a system of delay derivative equations. It is evident that a more comprehensive model of an epidemic should incorporate not only present-time information but also relevant data over past. Thus, the organisation of the paper is as follows. In section 2, we formulated a dynamical delayed-differential equation for effective dengue virus transmission between humans and vectors. In Section 3, we carried out preliminary results for a non-delayed model, such as positivity and boundedness of solution, non-dimensionalization, and equilibria points of the dimensionless model. Furthermore, in Section 4, we presented stability analysis for the non-delayed system using reproduction number as a threshold for the Dengue-free equilibrium point and Dengue-endemic equilibrium point, while in Section 5, we demonstrated stability analysis for the delayed system  and performed the results for the numerical computation analysis (i.e. periodic oscillations). Finally, in Section 6, we presented the conclusion and future directions.

\section{Model Formulation}
We are motivated by the work of \cite{garba2008backward}, which stratified the two-interacting population into the human population $(N_{h}(t))$ and the vector population $(N_{v}(t))$ to study the dengue infection dynamics in a given locality at any time $t$. We shall make a slight modification to the model in \cite{garba2008backward} by incorporating delay terms $\tau_{h}$ and $\tau_{v}$, which denote the meantime intrinsic and extrinsic incubations for humans and vectors, respectively while also modifying the forces of infection for both humans and vectors to reflect the delay terms. The total human population is divided into four categories based on disease status: susceptible $S_{h}(t)$, recovered $R_{h}(t)$, infectious $I_{h}(t)$ and exposed $E_{h}(t)$.
The total human population at time $t$ then defined as:
	$$N_{h}(t)=S_{h}(t)+E_{h}(t)+I_{h}(t)+R_{h}(t)$$
	Likewise, the vector population is divided into three categories: susceptible vectors $S_{v}(t)$, infectious vectors $I_{v}(t)$ and exposed vectors $E_{v}(t)$, resulting in a total vector population $N_{v}$ given as
	$$N_{v}(t)= S_{v}(t)+E_{v}(t)+I_{v}(t)$$
The mathematical formulation employed in the study of dengue fever dynamics is based on a nonlinear derivative deterministic system of equations:
\begin{equation}
\begin{aligned}
\frac{dS_{h}}{dt}&=\pi_{h} + \omega R_{h} - \lambda_{h}S_{h}-\mu_{h}S_{h},  \\
\frac{dE_{h}}{dt}&= \lambda_{h}S_{h} -(\mu_{h}+\sigma_{h})E_{h}, \\
\frac{dI_{h}}{dt}&=\sigma_{h}E_{h}-(\xi_{h}+\mu_{h}+\delta_{h})I_h,  \\
\frac{dR_{h}}{dt}&=\xi_{h}I_h-(\mu_{h}+\omega)R_{h},  \\
\frac{dS_{v}}{dt}&=\pi_{v}-\lambda_{v}S_{v}-(\mu_{v}+C_{v})S_{v},  \\
\frac{dE_{v}}{dt}&=\lambda_{v}S_{v}-(\theta_{c}+\sigma_{v}+\mu_{v}+C_{v})E_{v},  \\
\frac{dI_{v}}{dt}&=(\theta_{c}+\sigma_{v})E_{v}-(\mu_{v}+C_{v})I_{v},
\end{aligned}
\label{Equ:11_17}
\end{equation}
where $$\lambda_{h}=\frac{b\beta_{hv}I_{v}(t-\tau_{h})}{1+\nu_{h}I_{v}(t-\tau_{h})}, \qquad \lambda_{v}=\frac{b\beta_{vh}I_{h}(t-\tau_{v})}{1+\nu_{v}I_{h}(t-\tau_{v})}, \qquad I_h= I_{h}(t-\tau_r).$$
The emerging variables and parameters from the model are presented in Table \ref{Table1}.
\vskip0.4cm\noindent
\noindent The intrinsic $(\tau_{h})$ and extrinsic $(\tau_{v})$ delays refer to the periods associated with the virus's lifecycle in human hosts and mosquito vectors. These delays are critical for understanding how the disease spreads and predicting an outbreak's dynamics. The intrinsic delay refers to the time it takes for an infected person with dengue virus to become infectious to mosquitoes while the extrinsic delay refers to the time it takes for the dengue virus to become transmissible within a mosquito after the mosquito has bitten an infected human. Finally, the term $I_h(t-\tau_r)$ denotes the number of infectious individuals who commence the recovery process after a period of time $\tau_r$ \cite{gonccalves2011oscillations}.\\
\noindent It should be noted that the mean time extrinsic incubation, denoted by $\tau_{v}$, and the mean time intrinsic incubation, denoted by $\tau_{h}$, are to be understood as separate quantities. \\
The term $\lambda_{h}=\frac{b\beta_{hv}I_{v}(t-\tau_{h})}{1+\nu_{h}I_{v}(t-\tau_{h})}$ indicates that the infection time $t$ is induced by an infectious class $I_{v}(t)$, and the earlier infection of $\tau_h$ units in time. Accordingly, $\tau_h$ denotes average incubation time of the produced human antibody in reaction to the vector infection.
Similarly, the term $\lambda_{v}=\frac{b\beta_{vh}I_{h}(t-\tau_{v})}{1+\nu_{v}I_{h}(t-\tau_{v})}$ denotes the transmission time $t$ prompted in infectious class $I_{h}(t)$, and $\tau_v$ represents the earlier infected time unit.
Therefore, the term $\tau_v$ represents the average duration of the incubation period associated with vector-produced antibodies in response to human infection.

\begin{table}[htbp]
    \centering
    \caption{Description of the variables and parameters of the dengue fever model \eqref{Equ:11_17}.}
    \begin{tabularx}{\textwidth}{c | l}
    \hline  \hline
		Variable & Description\\
		\hline  \hline
		$S_{h}$ & Susceptible humans population\\
		$E_{h}$ & Exposed humans population\\
		$I_{h}$ & Infectious humans population\\
		$R_{h}$ & Recovered humans population\\
		$S_{v}$ & Susceptible vector population\\
		$E_{v}$ & Exposed vector population\\
		$I_{v}$ & Infectious vector population\\
		\hline \hline
		Parameter & Description\\
		\hline \hline
		$\pi_{h}, \pi_{v}$ & Humans and vector recruitment rate respectively\\
		$\beta_{hv}$ & Probability of human-to-vector transmission\\
		$\beta_{vh}$ & Probability of vector-to-human transmission\\
		$b$ & Vector rate of biting \\
		$\nu_{h}$ & Human antibody proportion in reaction to infection incidence\\& induced by vector \\
		$\nu_{v}$ & Vector antibody proportion in reaction to infection incidence\\& induced by human\\
		$\mu_{v}$, $\mu_{h}$ & Vector and human natural death rate \\
		$\delta_{h}$  & Disease induced death rate for humans\\
		$\sigma_{h}$ & Exposed to infectious module of disease rate of progression \\
		$\sigma_{v}$ & Exposed vector to infectious module of disease rate of progression\\
		$\xi_{h}$ & Recovery human infectious rate due to treatment\\
		$\omega$ & Human immunity losing per capita rate\\
		$C_{v}$ & Control effect rate of vector \\
		$\theta_{c}$ & Extrinsic incubation rate of vector\\
		\hline \hline
    \end{tabularx}
	\label{Table1}
\end{table}

\section{Preliminary results for non-delayed model}
\label{Sec:NoDelay}

\subsection{Positivity and boundedness of solutions}
\label{Sec:Positivity}

\begin{lemma}
	\label{positivitylemma}
	In the absence of delay, let the initial data for the dengue fever model \eqref{Equ:11_17}  be $D(0)$ $\geq 0$, where $D(t)=\left(S_{h}(t), E_{h}(t), I_{h}(t), R_{h}(t), S_{v}(t), E_{v}(t),\right.$ $ \left.I_{v}(t)\right)$. Then the solutions $D(t)$ of the model with non-negative initial data will remain non-negative for all time $t>0$.
\end{lemma}
\begin{proof}
	Suppose the initial data set of the dengue model $D(0)\ge 0$  such that $D(t)=\left(S_{h}(t), E_{h}(t), I_{h}(t), R_{h}(t), S_{v}(t), E_{v}(t),\right.$ $ \left.I_{v}(t)\right)$ is non-negative for all $t \textgreater 0$. We begin by showing the positivity of the solutions as follows:\\Taking from the first model equation \eqref{Equ:11_17} that
	\begin{eqnarray}
		\label{intS}
		\frac{\d S_{h}}{\d t}&=&\pi_{h} +\omega R_{h}-\lambda_{h}S_{h}-\mu_{h}S_{h}\geq \pi_{h}-\lambda_{h}S_{h}-\mu_{h}S_{h},
	\end{eqnarray}
	Using the method of integrating factor, Eq. \eqref{intS} takes the form  \\
	\begin{eqnarray}
		\frac{\d}{\d t}\left(S_{h}(t)\textrm{exp}\left[\mu_{h}t+\int_{0}^{t}\lambda_{h}(\zeta)\d \zeta\right]\right)\geq \pi_{h}\textrm{exp}\left[\mu_{h}t+\int_{0}^{t}\lambda_{h}(\zeta)\d \zeta\right], \notag
	\end{eqnarray}
	Hence,
	\begin{eqnarray}
		S_{h}(t_{s}) \textrm{exp}\left[\mu_{h}t+\int_{0}^{t}\lambda_{h}(\zeta)\d \zeta\right]-S_{h}(0)\geq \int_{0}^{t}\pi_{h}\left(\textrm{exp}\left[\mu_{h}x+\int_{0}^{x}\lambda_{h}(\zeta)\d \zeta\right]\right)\d x, \notag
	\end{eqnarray}
	so that,
	\begin{eqnarray}
		S_{h}(t)& \geq& S_{h}(0) \textrm{exp}\left[-\mu_{h}t-\int_{0}^{t}\lambda_{h}(\zeta)\d\zeta\right] \notag\\
		&+& \textrm{exp}\left[-\mu_{h}t-\int_{0}^{t}\lambda_{h}(\zeta)\d\zeta\right] \times \int_{0}^{t}\pi_{h}\left(\textrm{exp}\left[\mu_{h}x+\int_{0}^{x}\lambda_{h}(\zeta)\d \zeta\right]\right)\d x >0. \notag
	\end{eqnarray}
	Similarly, the remaining state variable $\left\{S_{h}(t), E_{h}(t), I_{h}(t),R_{h}(t),S_{v}(t),E_{v}(t),I_{v}(t)\right\} > 0$ at time $t>0$. Thus, the solutions to the $D(t)$ model \eqref{Equ:11_17} remain positive for all initial conditions non-negative .
\end{proof}

\subsection{Invariant region}
\label{Sec:Invariant}

Now, the dengue fever model \eqref{Equ:11_17} is biologically analyzed in the feasible region. Given the feasible region  $\mathcal{J}=\mathcal{J}_{h}\times \mathcal{J}_{v} \in \mathcal{R}_{+}^{4}\times \mathcal{R}_{+}^{3}$ with
$$\mathcal{J}_{h}=\left\{ (S_{h}, E_{h}, I_{h}, R_{h})\in \mathcal{R}_{+}^{4}: N_{h}\leq\frac{\pi_{h}}{\mu_{h}}\right \},$$
and
$$\mathcal{J}_{v}=\left\{(S_{v}, E_{v}, I_{v}) \in \mathcal{R}_{+}^{3}: N_{v}\leq\frac{\pi_{v}}{\mu_{v}+C_{v}} \right\}.$$

\begin{lemma}
	The feasible region $\mathcal{J} \subset \mathcal{R}_{+}^{7}$ is positively invariant for the model \eqref{Equ:11_17} with non-negative initial conditions in $\mathcal{R}_{+}^{7}$.
\end{lemma}
\begin{proof}
	The rate of change of the human and vector total population are given respectively by
	\begin{eqnarray}
		\label{human}
		\frac{\d N_{h}}{\d t}&=&\pi_{h}-\mu_{h}N_{h}(t)-\delta_{h}I_{h}(t) \leq \pi_{h}-\mu_{h}N_{h}(t),
		\end{eqnarray}
		and
		\begin{eqnarray}
			\label{vector}
		\frac{\d N_{v}}{\d t}&=&\pi_{v}-(\mu_{v}+C_v) N_{v}(t),
	\end{eqnarray}
where $N_{h}=S_{h}+E_{h}+I_{h}+R_{h}$ and $N_{v}=S_{v}+E_{v}+I_{v}$.
The above equation \eqref{human} and \eqref{vector} yields $N_{h}(t)\leq N_{h}(0)e^{-\mu_{h}t} +\frac{\pi_{h}}{\mu_{h}}\left(1-e^{-\mu_{h}t}\right)$ and $N_{v}(t)= N_{v}(0)e^{-(\mu_{v}+C_{v})t} +\frac{\pi_{v}}{\mu_{v}+C_{v}}\left(1-e^{-(\mu_{v}+C_{v})t}\right)$. It follows that $N_{h}(t) \to \frac{\pi_{h}}{\mu_{h}}$ and $N_v(t) \to \frac{\pi_{v}}{\mu_{v}+C_{v}}$ as $t \to \infty$. In particular, $N_{h}(t)\leq \frac{\pi_{h}}{\mu_{h}}$ and $N_{v}(t)\leq \frac{\pi_{v}}{\mu_{v}+C_{v}}$ if the respective total human and vector population at the initial time $N_{h}(0)\leq \frac{\pi_{h}}{\mu_{h}}$ and $N_{v}(0)\leq \frac{\pi_{v}}{\mu_{v}+C_{v}}$. Hence, the feasible region $\mathcal{J}$ is positively invariant.
\end{proof}
Thus, it is appropriate to dynamically examine the dengue infection formulation \eqref{Equ:11_17} in the bio-feasible region $\mathcal{J}$ in which the formulation is epidemiologically and mathematically well-defined \cite{ojo2017lyapunov,ojo2021modeling}.

\subsection{Non-dimensionalization}
\label{Sec:NonD}
Introducing dimensionless variables:
%\begin{align*}
%\tilde{S}_h & = \frac{S_h}{S_{h0}} \qquad \tilde{E}_h = \frac{E_h}{E_{h0}}, \qquad \tilde{I}_h = \frac{I_h}{I_{h0}}, \qquad \tilde{R}_h = \frac{R_h}{R_{h0}}, \\ \tilde{S}_v & = \frac{S_v}{S_{v0}}, \qquad \tilde{E}_v = \frac{E_v}{E_{v0}}, \qquad \tilde{I}_v = \frac{I_v}{I_{v0}}, \qquad \tilde{t}=\frac{t}{T},
%\end{align*}
\begin{align*}
\tilde{S}_h & = \frac{\mu_h}{\pi_h} S_h \qquad \tilde{E}_h = \frac{v_v \sigma_h}{\mu_h}E_h, \qquad \tilde{I}_h = v_v I_h, \qquad \tilde{R}_h = \frac{v_v \mu_h}{\xi_h}R_h, \\
\tilde{S}_v & = \frac{\mu_h}{\pi_v}S_v, \qquad \tilde{E}_v = \frac{v_h(\theta_c+\sigma_v)}{\mu_h} E_v, \qquad \tilde{I}_v =v_h I_v, \qquad \tilde{t}=\mu_h t,
\end{align*}
and new dimensionless parameters:
\begin{align*}
\tilde{\omega}& = \frac{\omega \xi_h}{\mu_h \pi_h v_v}, \quad \tilde{b}_{1h} = \frac{b \beta_{hv}}{\mu_h v_h}, \quad \tilde{\alpha}_{1} = \frac{\pi_h \sigma_h v_v}{\mu_h^2}, \quad \tilde{\sigma}_h = 1+\frac{\sigma_h}{\mu_h}, \\
\tilde{\xi}_n & =1+ \frac{\xi_h+\delta_h}{\mu_h},\quad  \tilde{\eta}_h = 1+\frac{\omega}{\mu_h},\quad \tilde{b}_{1v} = \frac{b \beta_{vh}}{v_v \mu_h}, \quad \tilde{c}_{1v} = 1+\frac{c_v}{\mu_h},\\ \tilde{\alpha}_2 &= \frac{\pi_h v_h (\theta_c+\sigma_v)}{\mu_h^2}, \quad \tilde{c}_{2v} = 1+\frac{\theta_c+\sigma_v+c_v}{\mu_h}.
\end{align*}
the dimensionless version of the non-delayed model is given in the following
%\begin{align}
%\frac{dS_{h}}{dt}&=1+ \omega R_{h} -b_{1h}\frac{ S_h I_v(t-\tau_h)}{1+I_v(t-\tau_h)}- S_{h}, \label{equ21} \\
%\frac{dE_{h}}{dt}&= \alpha_1 b_{1h}\frac{ S_h I_v(t-\tau_h) }{1+I_v(t-\tau_h)} - \sigma_h E_{h}, \qquad \sigma_h>1 \label{equ22}\\
%\frac{dI_{h}}{dt}&= E_{h}- \xi_{h} I_{h}(t-\tau_r), \qquad \xi_h>1 \label{equ23} \\
%\frac{dR_{h}}{dt}&= I_{h}(t-\tau_r)-\eta_{h} R_{h}, \qquad \eta_h>1 \label{equ24}  \\
%\frac{dS_{v}}{dt}&=1- b_{1v} \frac{ S_v I_h(t-\tau_v) }{1+I_h(t-\tau_v)} -c_{1v} S_{v},\qquad c_{1v}>1 \label{equ25} \\
%\frac{dE_{v}}{dt}&= \alpha_2 b_{1v} \frac{ S_v I_h(t-\tau_v) }{1+I_h(t-\tau_v)} - c_{2v}  E_{v}, \qquad c_{2v}>1 \label{equ26} \\
%\frac{dI_{v}}{dt}&= E_{v}- c_{1v} I_{v}, \label{equ27}
%\end{align}
\begin{align}
\frac{dS_{h}}{dt}&=1+ \omega R_{h} -\lambda_h S_h - S_{h}, \label{equ21} \\
\frac{dE_{h}}{dt}&= \alpha_1 \lambda_h S_h - \sigma_h E_{h}, \qquad \sigma_h>1 \label{equ22}\\
\frac{dI_{h}}{dt}&= E_{h}- \xi_{h} I_{h}, \qquad \xi_h>1 \label{equ23} \\
\frac{dR_{h}}{dt}&= I_{h}-\eta_{h} R_{h}, \qquad \eta_h>1 \label{equ24}  \\
\frac{dS_{v}}{dt}&=1- \lambda_v S_v  -c_{1v} S_{v},\qquad c_{1v}>1 \label{equ25} \\
\frac{dE_{v}}{dt}&= \alpha_2 \lambda_v S_v  - c_{2v}  E_{v}, \qquad c_{2v}>1 \label{equ26} \\
\frac{dI_{v}}{dt}&= E_{v}- c_{1v} I_{v}, \label{equ27}
\end{align}
where
\begin{equation}
\lambda_{h}=\frac{b_{1h} I_{v}(t-\tau_{h})}{1+I_{v}(t-\tau_{h})}, \qquad \lambda_{v}=\frac{b_{1v} I_{h}(t-\tau_{v})}{1+I_{h}(t-\tau_{v})}, \qquad I_h = I_{h}(t-\tau_r). \label{Equ:Dimless}
\end{equation}

Here $(\tilde{\cdot})$ is omitted for simplicity.

% In the rest of the paper, unless stated otherwise, parameters are fixed to $\omega=2, b_{1h}=4, \alpha_1=1, \sigma_h=1.1, \xi_h=1.2, \eta_h=2, b_{1v}=1.2, c_{1v}=1.2, \alpha_2=1.2, c_{2v}=1.4.$

\subsection{Equilibria of the dimensionless model}

The system \eqref{equ21}-\eqref{equ27} always has a disease-free equilibrium that is
$$D_{0} = \left(S_{ho}, E_{h0}, I_{h0}, R_{h0}, S_{v0}, E_{v0}, I_{v0}  \right)=\left(1, 0, 0, 0, \frac{1}{c_{1v}}, 0, 0 \right)$$
which corresponds to the disappearance of dengue disease.
One possible way to  understand more about the other roots of the equations
for equilibria points of the system \eqref{equ21}-\eqref{equ27} is to express all the variables in terms of $S_h$. First, we will derive the expressions for  these equations in this section. Using \eqref{equ21} and \eqref{equ22} gives
\begin{equation}
\frac{\sigma_h}{\alpha_1}E_h = 1+\omega R_h -S_h,\label{equ3}
\end{equation}
Note that Eqs. \eqref{equ23} and \eqref{equ24} lead to  $E_h = \xi_h I_h=\xi_h \eta_h R_h$. Substituting this in \eqref{equ3} gives the relation
\begin{equation}
R_h = \frac{\alpha_1 (1-S_h)}{(\sigma_h \xi_h \eta_h - \alpha_1 \omega)}.
\end{equation}
with $\sigma_h \xi_h \eta_h > \alpha_1 \omega$.
It is now straightforward to find the expressions for $I_h$ and $E_h$ in terms of $S_h$ as
\begin{equation}
 I_h = \frac{\eta_h \alpha_1 (1-S_h)}{(\sigma_h \xi_h \eta_h - \alpha_1 \omega)} \quad \textrm{and} \quad E_h = \frac{\xi_h \eta_h \alpha_1 (1-S_h)}{(\sigma_h \xi_h \eta_h - \alpha_1 \omega)}.
 \label{equ4}
\end{equation}
Similarly substituting $E_h$ given in  Eq. \eqref{equ4} into \eqref{equ22}:
\begin{equation}
I_v = \frac{\sigma_h \xi_h \eta_h (1-S_h)}{b_{1h} S_h (\sigma_h \xi_h \eta_h -\alpha_1 \omega) -\sigma_h \xi_h \eta_h (1-S_h)}.
\label{equ5}
\end{equation}
Since  $E_v= c_{1v} I_v$ in Eq. \eqref{equ27},  it is straightforward to obtain $E_v$ with respect to $S_h$. Besides, using Eq. \eqref{equ25} and the first expression of Eq.  \eqref{equ4}, susceptible vector population $S_v$ takes the form of susceptible human population as
\begin{equation}
S_v = \frac{\sigma_h \xi_h \eta_h -\alpha_1 \omega+\eta_h \alpha_1 (1-S_h)}{c_{1v} (\sigma_h \xi_h \eta_h -\alpha_1 \omega)+ \eta_h \alpha_1 (c_{1v}+b_{1v})(1-S_h) }.
\label{equ6}
\end{equation}
Lastly, substituting  the expressions for $S_v$ and $E_v$ in Eq.  \eqref{equ25}, the analytical expression for $S_h$ can be found in terms of system parameters using
\begin{equation}
S_h = 1- \left(\frac{u_2 u_4-u_1 u_5}{u_1 u_6-u_3 u_4}\right),
\label{equ7}
\end{equation}
where
\begin{align*}
u_1 =& \eta_h \alpha_1 b_{1v},\\
u_2 =& c_{1v} (\sigma_h\xi_h\eta_h-\alpha_1 \omega),\\
u_3=& \eta_h \alpha_1 (c_{1v}+b_{1v}),\\
u_4 = & c_{1v} c_{2v}\sigma_h \xi_h \eta_h,\\
u_5 =& \alpha_2 b_{1h} (\sigma_h\xi_h\eta_h-\alpha_1 \omega),\\
u_6 =& -\alpha_2 ( b_{1h} (\sigma_h\xi_h\eta_h-\alpha_1 \omega)+\sigma_h \xi_h \eta_h),
\end{align*}
One should note that all entrenched terms from the model are positive for their biological meaning, and hence Eq. \eqref{equ7} has a positive root if $u_4(u_2+u_3)<u_1(u_5+u_6)$.

\section{Stability Analysis for the non-delayed system using Reproduction Number threshold}

The disease threshold $R_0$ can be defined as the average number of dengue disease caused by a unitary infected vector or human when introduced into a population containing only susceptible individuals. The reproduction number can only be valid for dimensional models and real model parameters obtained from various statistical methods. By applying non-dimensionalization, we remove the dimensions and focus only on dynamics. Therefore, the reproduction number $R_0$  can be called differently.

\noindent Using the dimensionless model given by \eqref{equ21}-\eqref{equ27}, a relation between $\lambda_h$ and $\lambda_v$ can be obtained with
%\begin{eqnarray}
%\label{S8}
%S_{h} &=& \frac{\eta_{h}\sigma_{h}\xi_{h}}{\big(\eta_{h}\sigma_{h}\xi_{h}-\omega\alpha_{1}\big)\lambda_{h}+\eta_{h}\sigma_{h}\xi_{h}}, \quad
%E_{h}  = \frac{\eta_{h}\sigma_{h}\xi_{h}\alpha_{1}\lambda_{h}}{\big(\eta_{h}\sigma_{h}\xi_{h}-\omega\alpha_{1}\big)\lambda_{h}+\eta_{h}\sigma_{h}\xi_{h}} \notag\\
%I_{h}&=& \frac{\eta_{h}\alpha_{1}\lambda_{h}}{\big(\eta_{h}\sigma_{h}\xi_{h}-\omega\alpha_{1}\big)\lambda_{h}+\eta_{h}\sigma_{h}\xi_{h}}\quad
%R_{h} = \frac{\alpha_{1}\lambda_{h}}{\big(\eta_{h}\sigma_{h}\xi_{h}-\omega\alpha_{1}\big)\lambda_{h}+\eta_{h}\sigma_{h}\xi_{h}} \notag \\
%S_{v} &=& \frac{1}{\lambda_{v} + C_{1v}}, \quad
%E_{v} =  \frac{\alpha_{2} \lambda_{v}}{C_{2v}(\lambda_{v} + C_{1v})}, \quad
%I_{v} = \frac{\alpha_{2} \lambda_{v}}{C_{2v} C_{1v}(\lambda_{v} +C_{1v})}  \notag
%\end{eqnarray}
%using the \eqref{equ2}, we have
\begin{equation}
	\label{Equ:S9}
	\lambda_{v} = \frac{A_{1}\lambda_{h}}{A_{2}\lambda_{h} + \eta_{h}\sigma_{h}\xi_{h}}, \quad \textrm{or} \quad \lambda_{h} = \frac{A_{3}\lambda_{v}}{A_{4}\lambda_{v} + A_{5}},
\end{equation}\\
where $A_{1}= b_{1v}\alpha_{1}\eta_{h}$,
 $ A_{2} = \eta_{h}\sigma_{h}\xi_{h}+\eta_{h}\alpha_{1}-\omega\alpha_{1}$, $A_{3} = b_{1h} \alpha_{2}$, $A_{4} = c_{1v}c_{2v} + \alpha_{2}$, $ A_{5} = c_{2v} c_{1v}^2$.
This leads to
\begin{equation*}
(A_{1}A_{4} + A_{2}A_{5})\lambda_{h}^{2} + (\eta_{h}\sigma_{h}\xi_{h}A_{5} - A_{1}A_{3})\lambda_{h} = 0,
\end{equation*}
resulting in
\begin{equation*}
 \lambda_{h} = 0~~ or~~   \lambda_{h} = \frac{A_{1}A_{3} -\eta_{h}\sigma_{h}\xi_{h}A_{5}}{A_{1}A_{4}+ A_{2} A_{5}} > 0 ~~~
if  ~~R_{0} > 1.
\end{equation*}
Namely, one can conclude
\begin{equation}
\frac{A_{1}A_{3} -\eta_{h}\sigma_{h}\xi_{h}A_{5}}{A_{1}A_{4}+ A_{2} A_{5}} = \frac{\eta_{h}\sigma_{h}\xi_{h} A_{5}}{\big(A_{1}A_{4}+A_{2}A_{5}\big)}\left(\frac{A_{1} A_{3}}{\eta_{h}\sigma_{h}\xi_{h} A_{5}} -1 \right),
\label{Equ:R00}
\end{equation}
with
\begin{equation*}
\frac{\eta_{h}\sigma_{h}\xi_{h} A_{5}}{\big(A_{1}A_{4}+A_{2}A_{5}\big)}\left(R_{0}^{2} -1 \right) > 0, ~~  if~~ R_{0} > 1.
\end{equation*}
Thus,
\begin{equation*}
\lambda_{h} = 0, ~~or~~ \lambda_{h} = \frac{\eta_{h}\sigma_{h}\xi_{h} A_{5}}{\big(A_{1}A_{4}+A_{2}A_{5}\big)} (R_{0}^{2} - 1).
\end{equation*}
Note that,
\begin{equation*}
A_{1}A_{4}+A_{2}A_{5} =  A_{1}A_{4}+A_{5}(\eta_{h}\sigma_{h}\xi_{h} + \eta_{h}\alpha_{1} - \omega\alpha_{1}),
\end{equation*}
\begin{equation*}
 =  A_{1}A_{4}+A_{5}(\eta_{h}\sigma_{h}\xi_{h} + (\eta_{h} - \omega)\alpha_{1} ) > 0, ~~if ~~\eta_{h} \geq \omega.
\end{equation*}
The matrix $\mathcal{F}$ associated with new
infections and the matrix $\mathcal{V}$ containing the remaining expressions are given by\\
\\
$$\mathcal{F}=\left[\begin{array}{cccc}
 \alpha_{1} \lambda_h S_h
\\
0
\\
 \alpha_2 \lambda_v S_v
\\
 0
\end{array}\right]
 \quad \textrm{and} \quad
-\mathcal{V}=\left[\begin{array}{cccc}
 \sigma_h E_h
\\
-E_h+\xi_h I_h
\\
c_{2v} E_v
\\
 -E_v+c_{1v}I_v
\end{array}\right].$$
\\
Here, using the next-generation matrix technique, the transmission $(F)$ and transition $(V)$ matrices of the system \eqref{equ21}-\eqref{equ27} at dengue-free equilibrium are:

$$F|_{D0}=\left[\begin{array}{cccc}
0 & 0 & 0 & \alpha_{1} b_{1h}
\\
 0 & 0 & 0 & 0
\\
 0 & \frac{\alpha_{2} b_{1v}}{c_{1v}} & 0 & 0
\\
 0 & 0 & 0 & 0
\end{array}\right]
 \quad \textrm{and} \quad V|_{D0}=\left[\begin{array}{cccc}
\sigma_{h} & 0 & 0 & 0
\\
 -1 & \xi_{h} & 0 & 0
\\
 0 & 0 & c_{2 v} & 0
\\
 0 & 0 & -1 & c_{1v}
\end{array}\right].$$
\\
The the basic reproduction number $R_{0}$ is described  to be the spectral radius of the matrix $FV^{-1}$. Here the domainant eigenvalue is associated with the basic reproduction ratio that is \\
\begin{equation}
	R_{0}=\sqrt{\frac{A_{1} A_{3}}{\eta_{h}\sigma_{h}\xi_{h} A_{5}}}.
	\label{R0S}
\end{equation}
The Eq. \eqref{R0S} is described as the number of secondary infections caused by an infectious host and vector in the infection stage. It is important to note that infected vectors are produced by infected hosts, and vice-versa as a result of cross-infections between humans and vectors. The square root in Eq. \eqref{R0S} rises because it describes the number of expected geometric mean in human and vector secondary cases.

\noindent {\bf Case 1:} $\lambda_{h} = 0$ (Dengue-free equilibrium point that is  no infection transmission in the  population)

By substituting $\lambda_{h} = 0$ in Eq. \eqref{Equ:S9} gives $\lambda_{v} = 0$.
Thus,
\begin{equation*}
(\lambda_{h0}, \lambda_{v0}) = (0, 0), ~~ \textrm{corresponding to the DFE.}
\end{equation*}
Hence, the components of the DFE denoted by $D_{0}$ is

\begin{equation*}
D_{0} = \left(S_{ho}, E_{h0}, I_{h0}, R_{h0}, S_{v0}, E_{v0}, I_{v0}  \right),
\end{equation*}
 where
$$S_{ho} = 1,~ E_{h0} =0,~ I_{h0} = 0,~ R_{h0} = 0,~ S_{v0} = \frac{1}{c_{1v}},~ E_{v0} = 0,~ I_{v0} = 0.$$

\noindent {\bf Case 2:} $\lambda_{h} = \dfrac{\eta_{h}\sigma_{h}\xi_{h} A_{5}}{\big(A_{1}A_{4}+A_{2}A_{5}\big)}\left(R_{0}^{2} -1 \right)$  (Dengue-endemic  equilibrium point that is a state  where there is dengue transmission in both population)\\

By substituting $\lambda_{h} = \dfrac{\eta_{h}\sigma_{h}\xi_{h} A_{5}}{\big(A_{1}A_{4}+A_{2}A_{5}\big)}\left(R_{0}^{2} -1 \right)$ in \eqref{Equ:S9}, we have
\begin{equation*}
\label{S99}
\lambda_{v} = \frac{A_{1}\dfrac{\eta_{h}\sigma_{h}\xi_{h} A_{5}}{A_{1}A_{4}+A_{2}A_{5}}\left(R_{0}^{2} -1 \right)}{A_{2}\dfrac{\eta_{h}\sigma_{h}\xi_{h} A_{5}}{A_{1}A_{4}+A_{2}A_{5}}\left(R_{0}^{2} -1 \right) + \eta_{h}\sigma_{h}\xi_{h}}> 0, ~~if~~R_{0} > 1.
\end{equation*}
Thus,
\begin{equation*}
(\lambda_{h}^{*},\lambda_{v}^{*}) = \left( \frac{\eta_{h}\sigma_{h}\xi_{h} A_{5}}{A_{1}A_{4}+A_{2}A_{5}}\left(R_{0}^{2} -1 \right),   \frac{A_{1}\dfrac{\eta_{h}\sigma_{h}\xi_{h} A_{5}}{A_{1}A_{4}+A_{2}A_{5}}\left(R_{0}^{2} -1 \right)}{A_{2}\dfrac{\eta_{h}\sigma_{h}\xi_{h} A_{5}}{A_{1}A_{4}+A_{2}A_{5}}\left(R_{0}^{2} -1 \right) + \eta_{h}\sigma_{h}\xi_{h}} \right),
\end{equation*}
corresponds to the Dengue-endemic equilibrium, $ D^{*}$ with the following components:
\begin{equation*}
D^{*} =  \left(S_{h}^{*}, E_{h}^{*}, I_{h}^{*}, R_{h}^{*}, S_{v}^{*}, E_{v}^{*}, I_{v}^{*}  \right),
\end{equation*}
where
\begin{equation}
\begin{aligned}
S_{h}^{*} &= \frac{\eta_{h}\xi_{h}\sigma_{h}}{\big(\eta_{h}\sigma_{h}\xi_{h}-\omega\alpha_{1}\big)\lambda^{\star}+\eta_{h}\sigma_{h}\xi_{h}},~~ E_{h}^{*} = \frac{\eta_{h}\xi_{h}\alpha_{1}\lambda_{h}^{\star}}{\big(\eta_{h}\sigma_{h}\xi_{h}-\omega\alpha_{1}\big)\lambda^{\star}+\eta_{h}\sigma_{h}\xi_{h}},\\ I_{h}^{*} &= \frac{\eta_{h}\alpha_{1}\lambda_{h}^{\star}}{\big(\eta_{h}\sigma_{h}\xi_{h}-\omega\alpha_{1}\big)\lambda^{\star}+\eta_{h}\sigma_{h}\xi_{h}},~~
	R_{h}^{*} = \frac{\alpha_{1}\lambda_{h}^{\star}}{\big(\eta_{h}\sigma_{h}\xi_{h}-\omega\alpha_{1}\big)\lambda^{\star}+\eta_{h}\sigma_{h}\xi_{h}},\\
	S_{v}^{*} &= \frac{1}{\lambda_{v}^{*} + c_{1v}}, ~~E_{v}^{*} = \frac{\alpha_{2}\lambda_{v}^{*}}{c_{2v}(\lambda_{v}^{*} + c_{1v})},~~~ I_{v}^{*}= \frac{\alpha_{2} \lambda_{v}^{*}}{c_{2v}c_{1v}(\lambda_{v}^{*} + c_{1v})}. \nonumber
\end{aligned}
\end{equation}

\begin{theorem}
Consider the nondimensional system  \eqref{equ21}-\eqref{equ27}, there exists a unique positive endemic equilibrium if and only if $R_{0} > 1$, otherwise. The result follows that the DFE, $D_{0}$ is locally stable if $R_{0} < 1$ otherwise unstable.
\end{theorem}
Suppose that the stability of $D_{0}$ is independent of the initial size of the infected population, thus there is need to consider the global stability of $D_{0}$.
Consider the following Lyapunov function with carefully chosen scalar quantities
\begin{equation}
	\label{S11}
L(E_{h}, I_{h}, E_{v}, I_{v}) = \alpha_{2} b_{1v}E_{h} + \sigma_{h} \alpha_{2} b_{1v}I_{h} + \sigma_{h}\xi_{h}c_{1v}E_{v} + \sigma_{h}\xi_{h}c_{1v}c_{2v}I_{v}.
\end{equation}
The differentiation of \eqref{S11} with respect to $t$ along the system solutions \eqref{Equ:Dimless}, yields
\begin{equation*}
\begin{aligned}
    L^{'} =\alpha_{2} b_{1v}E^{'}_{h} + \sigma_{h} \alpha_{2} b_{1v}I^{'}_{h} + \sigma_{h}\xi_{h}c_{1v}E{'}_{v} + \sigma_{h}\xi_{h} c_{1v} c_{2v}I^{'}_{v}
\end{aligned}
\end{equation*}
\begin{equation*}
    \begin{aligned}
     L^{'} &=\alpha_{2} b_{1v}\big(\alpha_{1} \lambda_{h}S_{h}-\sigma_{h}E_{h}\big) + \sigma_{h} \alpha_{2} b_{1v}\big(E_{h}-\xi_{h} I_{h}\big) \\
     &+\sigma_{h}\xi_{h}c_{1v}\big(\alpha_{2} \lambda_{v} S_{v}-c_{2v}E_{v}\big) + \sigma_{h}\xi_{h}c_{1v} c_{2v}\big(E_{v}-c_{1v} I_{v}\big),
    \end{aligned}
\end{equation*}
leading to
\begin{equation*}
\begin{aligned}
    L^{'} &= \frac{\alpha_{2} b_{1v} \alpha_{1} b_{1h}I_{v}}{1 + I_{v}}S_{h}    + \frac{\alpha_{2}b_{1v}\xi_{h} \sigma_{h} c_{1v}I_{h}}{1 + I_{h}}S_{v} - \sigma_{h}\alpha_{2}\xi_{h}b_{1v}I_{h}- \xi_{h}\sigma_{h}c_{2v}c^{2}_{1v}I_{v}.
\end{aligned}
\end{equation*}
Further simplification yields,
\begin{equation*}
L^{'} =\left( \alpha_{2} b_{1v} \alpha_{1} b_{1h}S_{h} - \xi_{h}\sigma_{h}c_{2v}c^{2}_{1v}\right)I_{v} - \frac{\alpha_{2}b_{1v}\alpha_1 b_{1h}I^{2}_{v}}{(1 + I_{v})}S_{h}- \frac{\xi_{h}\sigma_{h}\alpha_{2}b_{1v} c_{1v} I^{2}_{h}}{(1 + I_{h})}S_{v}.
\end{equation*}
At the disease free equilibrium, we have
\begin{equation*}
S_{h} = 1,~~~S_{v} = \frac{1}{c_{1v}}.
\end{equation*}
%\textcolor{red}{Please add some connecting sentences between mathematical expressions.}
Then $L^{'}$ reduced to
\begin{equation*}
L^{'} = \xi_{h}\sigma_{h}c_{2v}c^{2}_{1v}\left(\frac{\alpha_{1}\alpha_{2}b_{1v}b_{1h}}{\xi_{h}\sigma_{h}c_{2v}c^{2}_{1v}} - 1 \right) I_{v} - \frac{\alpha_1 \alpha_{2}b_{1v}b_{1h}I^{2}_{v}}{(1 + I_{v})}- \frac{\xi_{h}\sigma_{h}\alpha_{2}b_{1v}I^{2}_{h}}{(1 + I_{h})}.
\end{equation*}
Therefore,
\begin{equation*}
L^{'} = \xi_{h}\sigma_{h}c_{2v}c^{2}_{1v}\left(R^{2}_{0} - 1 \right) I_{v} - \frac{\alpha_1 \alpha_{2}b_{1v}b_{1h}I^{2}_{v}}{(1 + I_{v})}- \frac{\xi_{h}\sigma_{h}\alpha_{2}b_{1v}I^{2}_{h}}{(1 + I_{h})} < 0,
\end{equation*}
when $ R_{0} \leq 1$.

\noindent Thus, the DFE is globally asymptotically stable. In the case where conditions that favor. Dengue disease to thrive, then the DFE is no longer globally stable, leading to an endemic state. It is desirable to study the stability analysis of the endemic equilibrium. The eigenvalues using the upper triangular matrix of the Jacobian matrix of the system \eqref{equ21}-\eqref{equ27} at the endemic equilibrium $D^{*}$ are given by
\begin{equation*}
\begin{aligned}
\lambda_{1} &= - \left(1 + \frac{b_{1h}I^{*}_{v}}{1+I^{*}_{v}}\right),\quad \lambda_{2} = - \sigma_{h}, \quad \lambda_{3} = - \xi_{h},\quad \lambda_{4} = - \eta_{h},\quad  \\ \lambda_{5} &= -A = -\left(c_{1v} + \frac{b_{1v}I^{*}_{h}}{1 + I^{*}_{h}}\right), \quad \lambda_{6} = -c_{1v}, \quad \lambda_{7} = -c_{1v}.
\end{aligned}
\end{equation*}
%\textcolor{magenta}{COULD YOU PLEASE CHECK THE FOLLOWING!!}
Clearly, both $\lambda_{1}$ and $\lambda_{5}$ are negative if $R_{0} > 1$. Therefore, the Dengue-endemic equilibrium is locally asymptotically stable.

\section{Stability Analysis for the delayed system}
\label{Sec:Stability}

The model \eqref{equ21}-\eqref{equ27} can be rewritten in a closed form:
\begin{equation}
\frac{\d}{\d t} \mathcal{V}(t) =H(\mathcal{V}(t), \mathcal{V}(t-\tau_h),  \mathcal{V}(t-\tau_v), \mathcal{V}(t-\tau_r), P),
\end{equation}
where $\mathcal{V}=(S_h, E_h,I_h,R_h,S_v,E_v,I_v)^T$  and $H: \mathbb{R}^{7 \times 7} \times \mathbb{R}^p \to \mathbb{R}^7 $ is a nonlinear function with a number of parameters $P \in \mathbb{R}^p$ and constant time delays $\tau_h$, $\tau_v$ and $\tau_r$. Using the linearization argument, the corresponding system at an equilibrium $\Sigma^*= \left(S_h^*, E_h^*, I_h^*, R_h^*, S_v^*, E_v^*, I_v^*\right)$ can be given as
\begin{equation}
\frac{\d}{\d t} \mathcal{Z}(t) = J_1^0  \mathcal{Z}(t) +\sum\limits_{i} J_{\tau_i}^D  \mathcal{Z}(t-\tau_i), \qquad i=\{h,v,r\}.
\end{equation}
Here, $\mathcal{Z}=(S_h', E_h', I_h', R_h', S_v', E_v', I_v')^T$ where the variables with primes represent perturbed variables for the linearization, e.g. for the susceptible human population we consider $S_h = S_h^* + S_h'$. The matrices are given as
\[J_0^D   = \left. \left(\frac{\partial H}{\partial \mathcal{V}}\right)\right|_{\Sigma^*}, \quad J_{\tau_h}^D   = \left. \left(\frac{\partial H}{\partial \mathcal{V}(t-\tau_h)}\right)\right|_{\Sigma^*},\]\[ J_{\tau_v}^D   = \left. \left(\frac{\partial H}{\partial \mathcal{V}(t-\tau_v)}\right)\right|_{\Sigma^*},\quad J_{\tau_r}^D   = \left. \left(\frac{\partial H}{\partial \mathcal{V}(t-\tau_r)}\right)\right|_{\Sigma^*}.\]
Defining a $7\times 7$ characteristic  matrix
\begin{equation}
J_\tau= J_0^D + \sum\limits_{i} J_{\tau_i}^D  \e^{-\lambda \tau_i}, \qquad i=\{h,v,r\}.
\label{Equ:Mat}
\end{equation}
Then, the corresponding characteristic equation is found by using
\begin{equation}
 \Psi(\lambda) = \textrm{Det} \left( \lambda I_7 - J_\tau\right)=0.
\label{Equ:Det}
\end{equation}
Solution of Eq. \eqref{Equ:Det} gives an infinite number of eigenvalues which determine the stability of the equilibrium. Since eigenvalues sufficiently close to the imaginary axis are enough to analyze stability, a finite number of solutions, where $\textrm{Real}(\lambda)>r,\; r \in \mathbb{R}^+$ can be taken into consideration.

\subsection{Endemic equilibrium}
\label{Sec:EndEqu}
We now focus on the endemic equilibrium given by $\Sigma_1^* = \left(S_h^*, E_h^*, I_h^*, R_h^*, S_v^*, E_v^*, I_v^*\right)$.
Using the linearization argument the explicit form of the matrix given in Eq.  \eqref{Equ:Mat} is given as follows
\[
J_\tau=\begin{pmatrix}
-m-1&0 & 0 & \omega & 0 & 0 &- f\e^{-\lambda \tau_h}\\
\alpha_1 m & -\sigma_h& 0 &0&0&0&\alpha_1 f \e^{-\lambda \tau_h}\\
0 & 1 & -\xi_h \e^{-\lambda \tau_r}  & 0&0&0&0\\
0&0& \e^{-\lambda \tau_r} & -\eta_h&0 &0&0\\
0&0&-h \e^{-\lambda \tau_v}& 0 & -g -c_{1v}&0&0\\
0&0& \alpha_2 h \e^{-\lambda \tau_v}&0&\alpha_2 g& -c_{2v}&0\\
0&0&0&0&0&1&-c_{1v}
\end{pmatrix}.
\]
Therefore, Eq. \eqref{Equ:Det} becomes
\begin{equation}
\Psi_{\Sigma_1^*}(\lambda) =  (\lambda+c_{1v}) \Omega_{\Sigma_1^*}( \lambda,\tau_h, \tau_v,\tau_r),
\label{Equ:Psii1ii}
\end{equation}
where
%\begin{equation}
%\begin{aligned}
%\Omega_{\Sigma_1^*}(\lambda,\tau_h, \tau_v, \tau_r) =& (\lambda+m+1)(\lambda+c_{1v}+g)(\lambda+\eta_h)(\lambda+c_{2v})(\lambda+\sigma_h) (\lambda+\xi_h \e^{-\lambda \tau_r}) \\ & -m \alpha_1\omega (\lambda+c_{2v})(\lambda+c_{1v}+g) \e^{-\lambda \tau_r}\\&
%-\alpha_1 \alpha_2 f h (\lambda+\eta_h)(\lambda+1)\e^{-\lambda (\tau_h+\tau_v)} ,
%\end{aligned}
%\end{equation}
%which can be rewritten as
\begin{equation}
\begin{aligned}
\Omega_{\Sigma_1^*}(\lambda,\tau_h, \tau_v,\tau_r) =& \lambda (\lambda+m+1)(\lambda+c_{1v}+g)(\lambda+c_{2v})(\lambda+\sigma_h) (\lambda+\eta_h ) \\ & \e^{-\lambda \tau_r} (\lambda+c_{2v}) (\lambda+c_{1v}+g) \left[\xi_h  (\lambda+\eta_h ) (\lambda+m+1)  (\lambda+\sigma_h ) \right. \\& \left.-m \alpha_1 \omega  \right]
-\alpha_1 \alpha_2 f h \e^{- \lambda (\tau_h+\tau_v) } (\lambda+1)(\lambda+\eta_h),
\end{aligned}
\end{equation}
by which Eq.~\eqref{Equ:Psii1ii}  leads to a seventh order equation of
\begin{equation}
\begin{aligned}
& \lambda \left[ \lambda^6 + a_1 \lambda^5 + a_2 \lambda^4 + a_3 \lambda^3 + a_4 \lambda^2 + a_5 \lambda+a_6\right]+ \e^{-\lambda \tau_r} \left[ \xi_h\left( \lambda^6 + a_1 \lambda^5 \right. \right. \\&  \left. \left.+ a_2 \lambda^4  + a_3 \lambda^3 + a_4 \lambda^2 + a_5 \lambda + a_6\right) - m \alpha_1 \omega \left(\lambda^3+b_1 \lambda^2 + b_2 \lambda +b_3 \right) \right] \\&+ \e^{- \lambda (\tau_h+\tau_v) }  \left[-\alpha_1 \alpha_2 f h  \left(\lambda^3+c_1 \lambda^2 + c_2 \lambda +c_3\right)  \right] =0,
\end{aligned}
\label{Equ:Chactaui}
\end{equation}
where
%a1= 2*c1v + c2v + etah + g + m + sigmah + 1;
%a2= c1v*c2v + (m + 1)*(c1v + c2v + etah + sigmah) + sigmah*(c1v + c2v + etah) + etah*(c1v + c2v) + (c1v + g)*(c1v + c2v + etah + m + sigmah + 1);
%a3 = sigmah*(c1v*c2v + etah*(c1v + c2v)) + (c1v + g)*(c1v*c2v + (m + 1)*(c1v + c2v + etah + sigmah) + sigmah*(c1v + c2v + etah) + etah*(c1v + c2v)) + (m + 1)*(c1v*c2v + sigmah*(c1v + c2v + etah) + etah*(c1v + c2v)) + c1v*c2v*etah;
%a4 =(sigmah*(c1v*c2v + etah*(c1v + c2v)) + c1v*c2v*etah)*(m + 1) + (c1v + g)*(sigmah*(c1v*c2v + etah*(c1v + c2v)) + (m + 1)*(c1v*c2v + sigmah*(c1v + c2v + etah) + etah*(c1v + c2v)) + c1v*c2v*etah) + c1v*c2v*etah*sigmah;
%a5 =(c1v + g)*((sigmah*(c1v*c2v + etah*(c1v + c2v)) + c1v*c2v*etah)*(m + 1) + c1v*c2v*etah*sigmah) + c1v*c2v*etah*sigmah*(m + 1);
%a6 = c1v*c2v*etah*sigmah*(c1v + g)*(m + 1);
\begin{align}
a_1= &2 c_{1v} + c_{2v} + \eta_h + g + m + \sigma_h + 1, \nonumber \\
a_2= & c_{1v} c_{2v} + (m + 1)(c_{1v} + c_{2v} + \eta_h + \sigma_h) + \sigma_h (c_{1v} + c_{2v} + \eta_h) \nonumber \\&+ \eta_h (c_{1v} + c_{2v}) + (c_{1v} + g) (c_{1v} + c_{2v} + \eta_h + m + \sigma_h + 1), \nonumber\\
a_3 =& \sigma_h (c_{1v} c_{2v} + \eta_h (c_{1v} + c_{2v})) + (c_{1v} + g)(c_{1v} c_{2v} + (m + 1) (c_{1v} \nonumber \\&+ c_{2v} + \eta_h + \sigma_h) + \sigma_h (c_{1v}+ c_{2v} + \eta_h) + \eta_h (c_{1v} + c_{2v})) \nonumber \\& + (m + 1)(c_{1v} c_{2v} + \sigma_h (c{1v} + c_{2v} + \eta_h)  + \eta_h (c_{1v} + c_{2v})) + c_{1v} c_{2v} \eta_h, \nonumber   \\
a_4 =& (\sigma_h (c_{1v} c_{2v} + \eta_h (c_{1v} + c_{2v})) + c_{1v} c_{2v} \eta_h)(m + 1) + (c_{1v} + g) (\sigma_h (c_{1v} c_{2v} \nonumber \\&+ \eta_h (c_{1v} + c_{2v})) + (m + 1)(c_{1v} c_{2v} + \sigma_h (c_{1v} + c_{2v} + \eta_h) \label{Equ:aii}  \\& + \eta_h (c_{1v} + c_{2v})) + c_{1v} c_{2v} \eta_h)  + c_{1v} c_{2v} \eta_h \sigma_h, \nonumber \\
a_5 =& (c_{1v} + g)((\sigma_h (c_{1v}c_{2v} + \eta_h (c_{1v} + c_{2v})) + c_{1v} c_{2v} \eta_h) (m + 1) \nonumber \\&+ c_{1v} c_{2v} \eta_h \sigma_h) + c_{1v} c_{2v} \eta_h \sigma_h (m + 1),  \nonumber\\
a_6  =&  c_{1v} c_{2v} \eta_h \sigma_h (c_{1v} + g) (m + 1), \nonumber\\
b_1= &2 c_{1v} + c_{2v} + g, \quad b_2= c_{1v} c_{2v} + (c_{1v} + c_{2v}) (c_{1v} + g), \quad b_3=c_{1v} c_{2v} (c_{1v }+ g), \nonumber \\
c_1=&c_{1v} + \eta_h + 1, \quad  c_2=c_{1v} + \eta_h + c_{1v} \eta_h, \quad c_3=c_{1v} \eta_h, \nonumber
\end{align}
with
\begin{equation}
m=\frac{b_{1h}I_v^*}{1+I_v^*}, \quad f= \frac{b_{1h} S_h^*}{(1+I_v^*)^2}, \quad h= \frac{b_{1v} S_v^*}{(1+I_h^*)^2}, \quad g=\frac{b_{1v} I_h^*}{1+I_h^*}.
\label{Equ:b1hhi}
\end{equation}
It can be easily observed from Eq.  \eqref{Equ:Psii1ii} that there is one negative  eigenvalue, that is $\lambda_1=-c_{1v}$. Note that the values of $\tau_h$ and $\tau_v$ given in Eq. \eqref{Equ:Chactaui}  do not independently affect the model. In fact, the sum of the delay parameters has a role in the system dynamics.
For simplicity we consider $\tau_h=\tau_v=\tau$.

We now show some numerical results complementing the theoretical formulas.
 In the rest of the paper, unless stated otherwise, parameters are fixed to $\omega=2, b_{1h}=4, \alpha_1=1, \sigma_h=1.1, \xi_h=1.2, \eta_h=2, b_{1v}=1.2, c_{1v}=1.2, \alpha_2=1.2, c_{2v}=1.4.$

\noindent In Fig.~\ref{Fig:TE1}, the temporal dynamics of all seven populations is presented for two different delay cases. In the first case, the effect of the change in the extrinsic and intrinsic incubation delay  is shown in Fig.~\ref{Fig:TE1}(a) and Fig.~\ref{Fig:TE1}(b), where  it is assumed that infectious individuals begin to recover without delay, e.g. $\tau_r=0$. Fig.~\ref{Fig:TE1}(a) represents the stable dynamics of the system \eqref{equ21}-\eqref{equ27} in the absence of delay. Increasing incubation delays from $\tau=0$ to $\tau=5$ ($\tau_h=\tau_v=\tau$) in Fig.~\ref{Fig:TE1}(b), all variables stabilise after some initial oscillatory dynamics  ended at around $t=100.$ For mathematical convenience, in Fig.~\ref{Fig:TE1}(c) and Fig.~\ref{Fig:TE1}(d), intrinsic and extrinsic incubation  time is chosen to be half of the recovery delay, e.g. ($\tau_h=\tau_v=\tau$) and $\tau_r=2 \tau$. Considering $\tau=0.55$ ($\tau_r=1.1$) in Fig.~\ref{Fig:TE1}(c) damping oscillations are observed leading to stability and   high frequency periodic oscillations with a Hopf bifurcation is observed with $\tau=0.6$ ($\tau_r=1.2$).

\begin{figure}[ht!]
\centering
\begin{tabular}{cccc}
\includegraphics[scale=0.28]{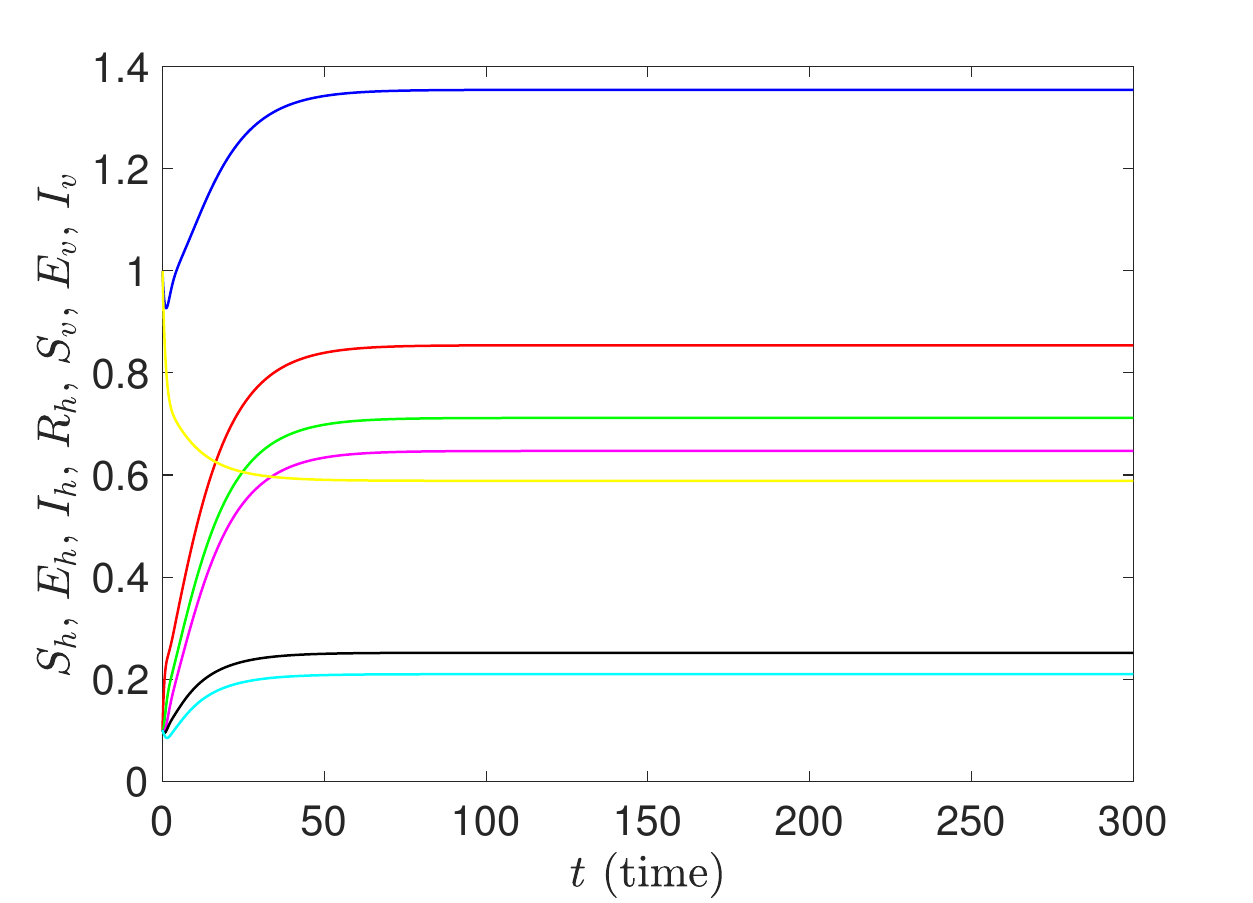} \\
\textbf{(a)}  \\[6pt]
\end{tabular}
\hspace{-1.03cm}
\vspace{-0.3cm}
\begin{tabular}{cccc}
\includegraphics[scale=0.28]{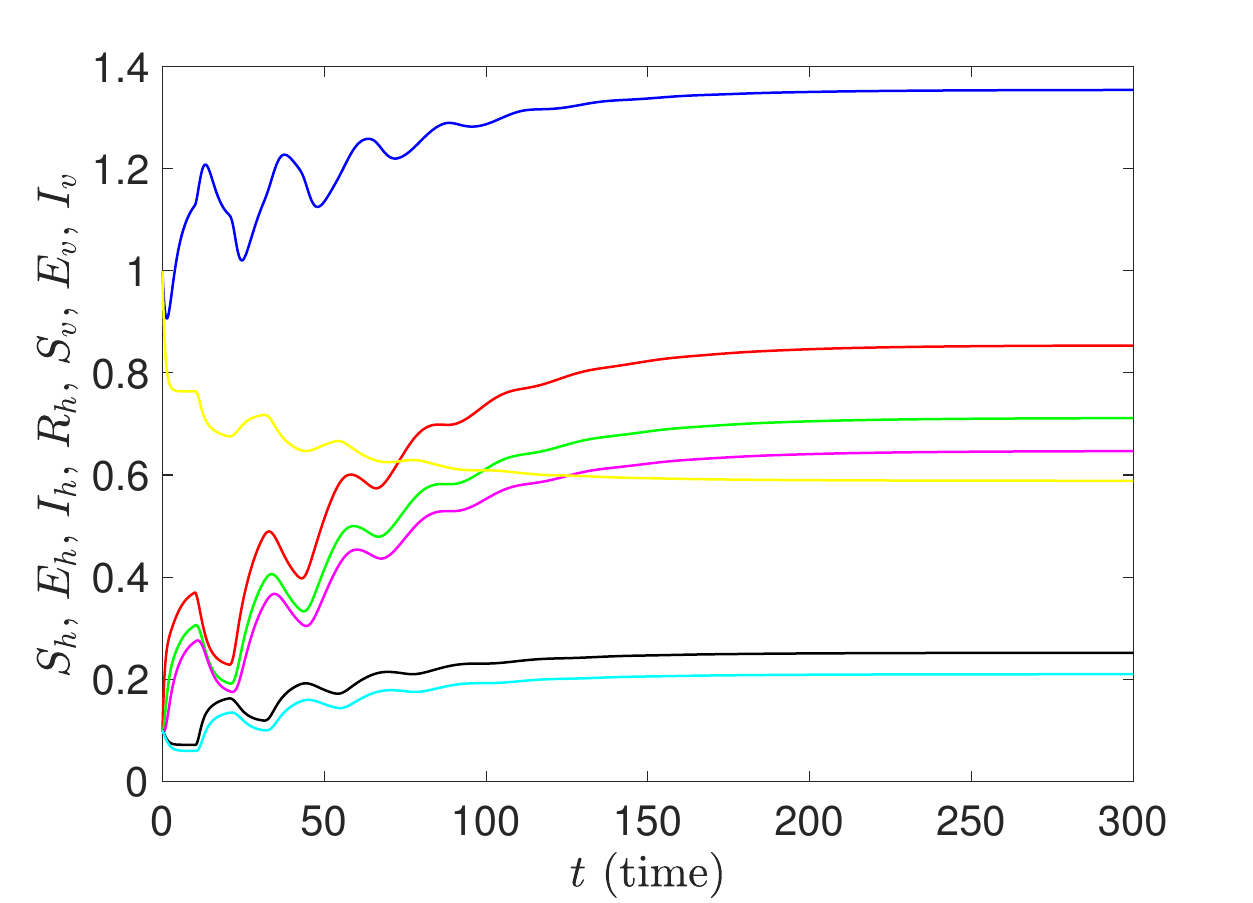} \\
\textbf{(b)}\\[6pt]
\end{tabular}
\begin{tabular}{cccc}
\includegraphics[scale=0.28]{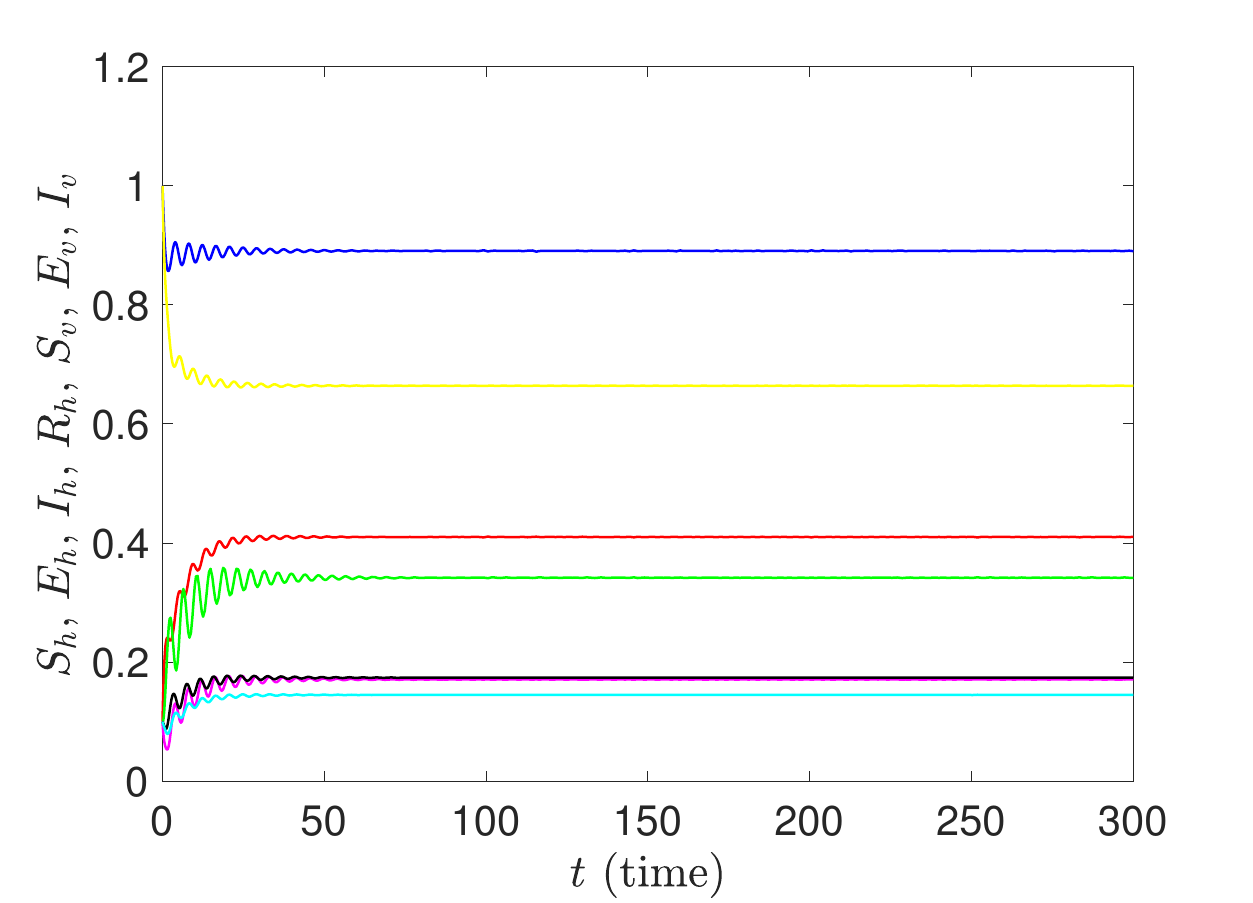} \\
\textbf{(c)}\\[6pt]
\end{tabular}
\hspace{-1.03cm}
\vspace{-0.2cm}
\begin{tabular}{cccc}
\includegraphics[scale=0.28]{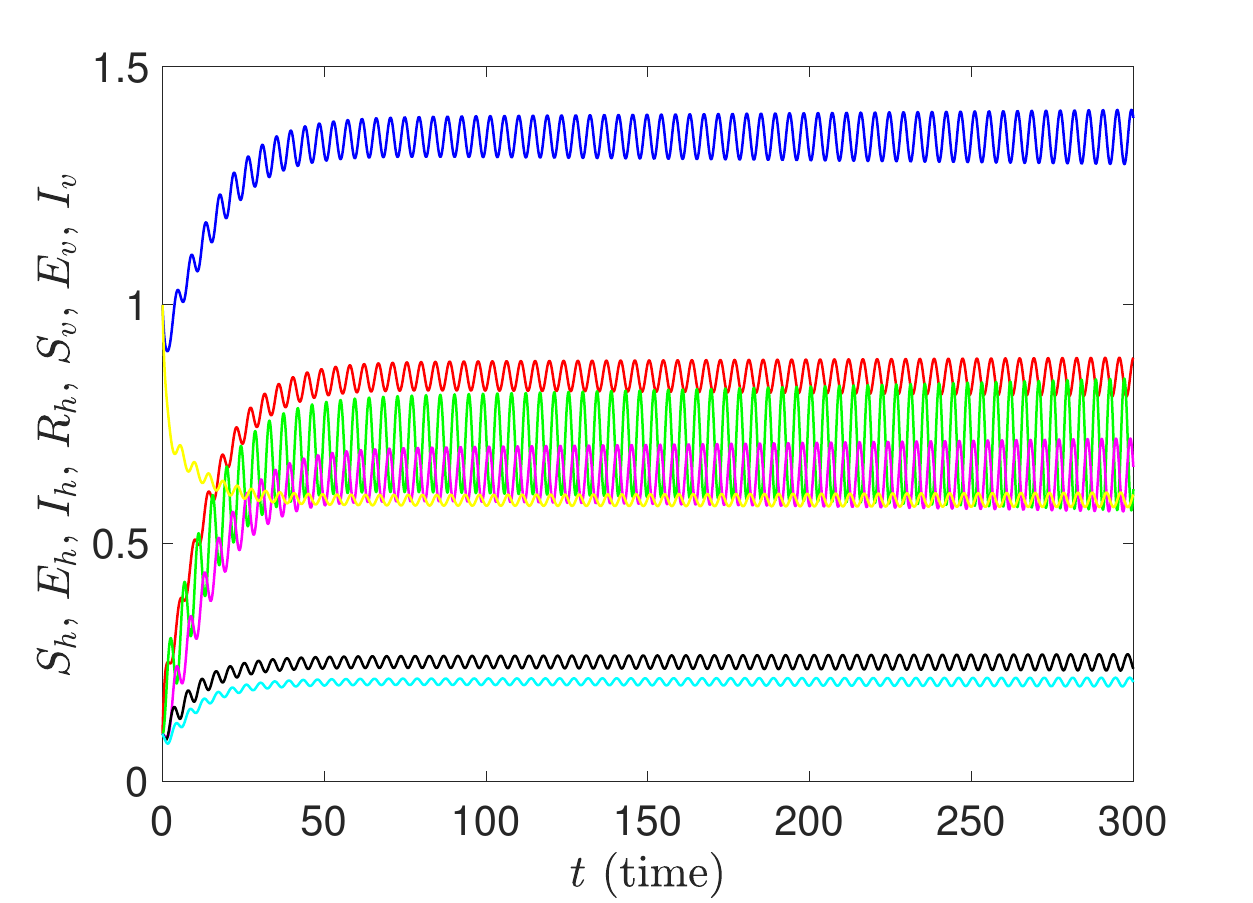} \\
\textbf{(d)}\\[6pt]
\end{tabular}
\begin{tabular}{cccc}
\includegraphics[scale=0.26]{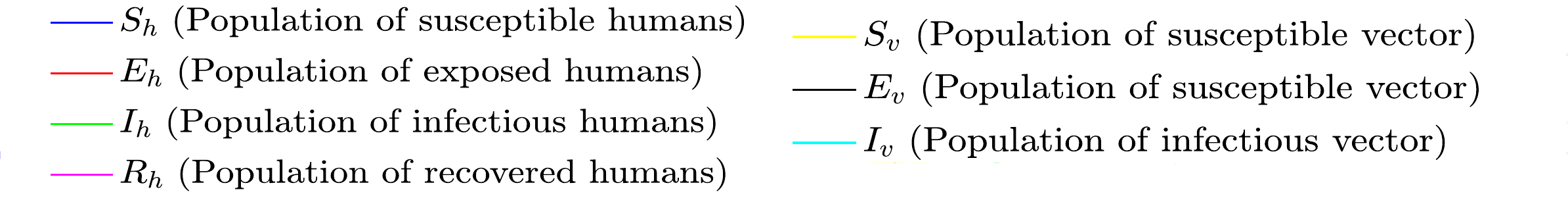}
\end{tabular}
\caption{Time evolution of each variable in the model \eqref{equ21}-\eqref{equ27}  in the absence of delay terms, e.g.  $\tau_r=\tau_h=\tau_v=0$ (a)  in the presence of average extrinsic and intrinsic incubation time with  $\tau_r=0$ and $\tau_h=\tau_v=5$ (b), in the presence of delay terms with $\tau_r=1.1$, $\tau_h=\tau_v=0.55$ (c) and $\tau_r=1.2$, $\tau_h=\tau_v=0.6$ (d). In the presence of all three delay terms it is assumed that $\tau_r=2 \tau_h=2 \tau_v$. }
\label{Fig:TE1}
\end{figure}

%the infectious individuals begin to recover after some time.

\noindent Figure \ref{Fig:TE2} shows the system dynamics in the absence of intrinsic and extrinsic incubation time ($\tau_h=\tau_v=0$) and it is assumed that the infectious individuals start to recover after some delay $\tau_r \neq 0$. If the recovery delay is lower than some threshold,  all variables saturates to their steady states with damping oscillations. Then a Hopf bifurcation  at a critical value of delay, e.g. $\tau_r^c \approx 1.4472$,  is encountered leading to a limit cycle surrounding the positive coexistence state. In Fig.~\ref{Fig:TE2}(b) dynamics with small amplitude oscillations is shown with $\tau_r=1.7$.  Further away from the Hopf point, e.g. $\tau_r=2.2$, dynamics for all variables exhibits sustainable periodic oscillations and the size of the limit cycle is substantially increased. If there is an excessive increase in the recovery delay with $\tau_r=2.4$, periodic oscillations with a stable limit cycle translate into chaos after some time, see Fig.~\ref{Fig:TE2}(d).

\begin{figure}[ht!]
\centering
\begin{tabular}{cccc}
\includegraphics[scale=0.28]{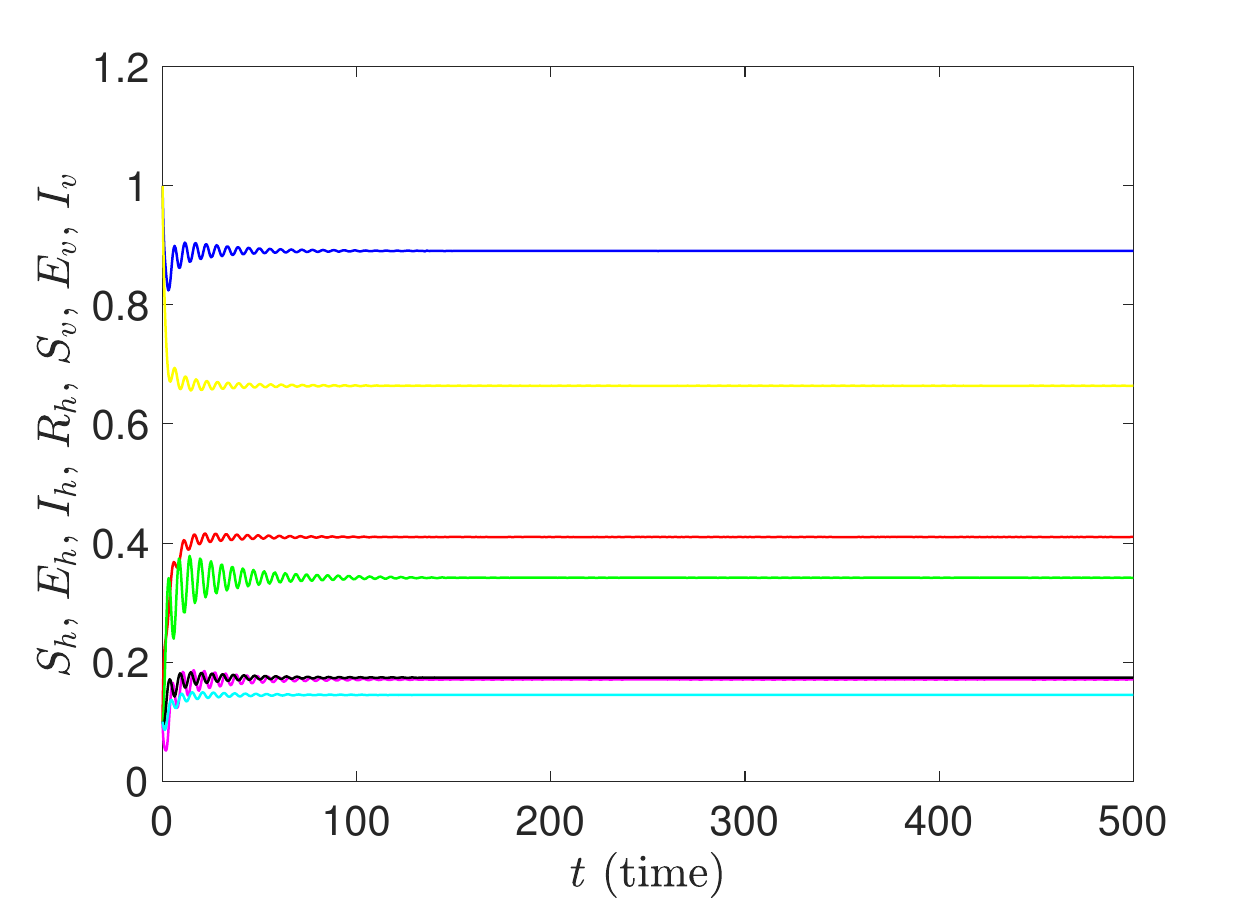} \\
\textbf{(a)}  \\[6pt]
\end{tabular}
\hspace{-1.03cm}
\vspace{-0.3cm}
\begin{tabular}{cccc}
\includegraphics[scale=0.28]{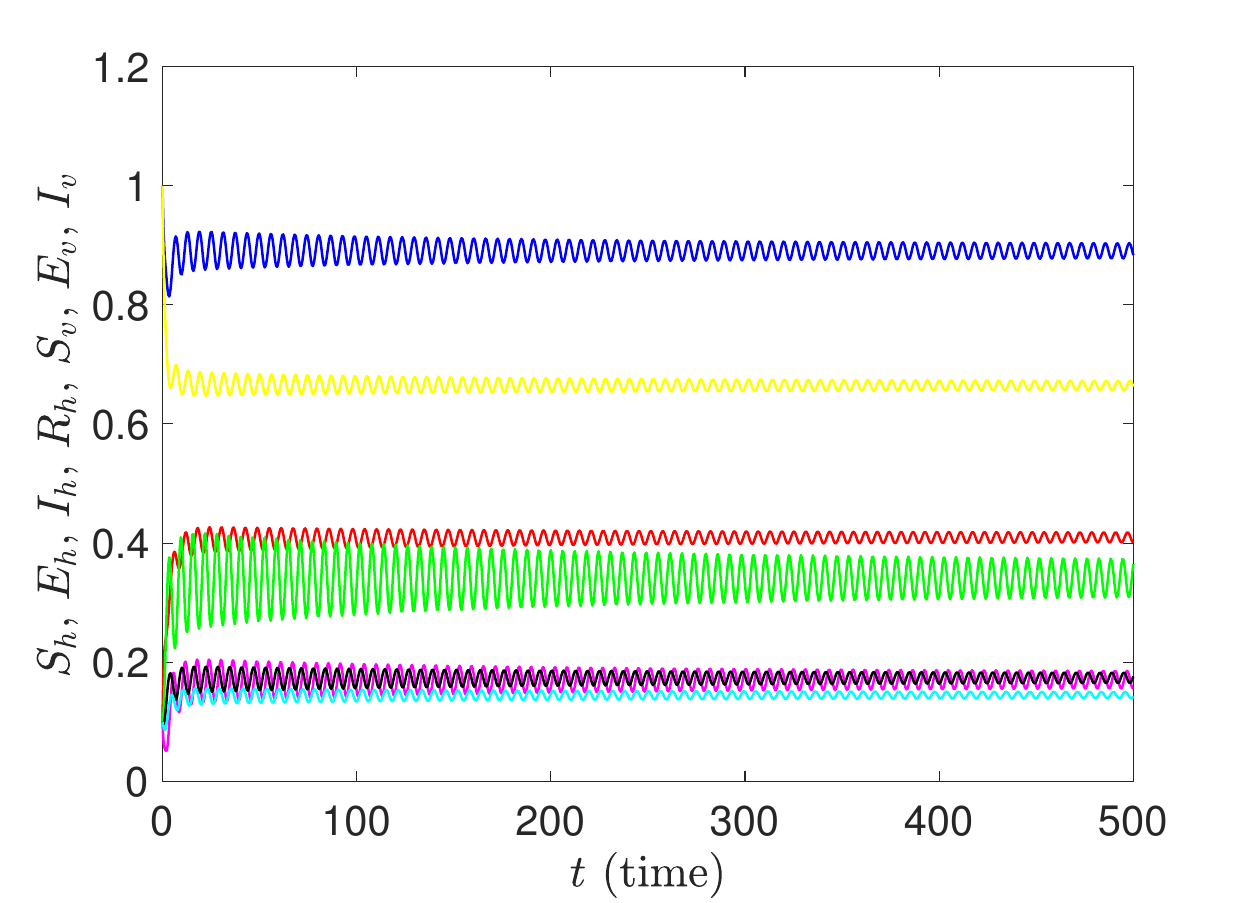} \\
\textbf{(b)}\\[6pt]
\end{tabular}
\begin{tabular}{cccc}
\includegraphics[scale=0.28]{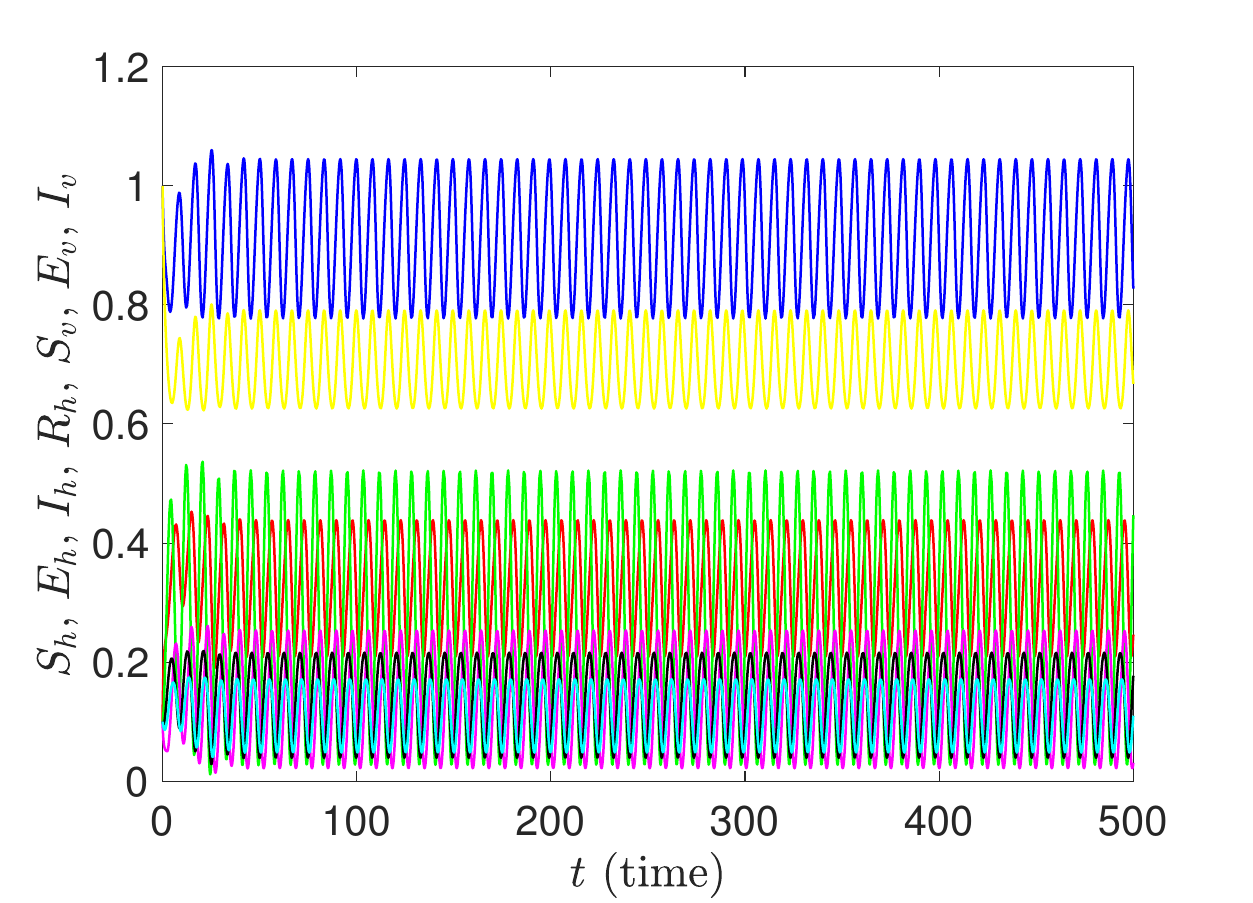} \\
\textbf{(c)}\\[6pt]
\end{tabular}
\hspace{-1.03cm}
\vspace{-0.2cm}
\begin{tabular}{cccc}
\includegraphics[scale=0.28]{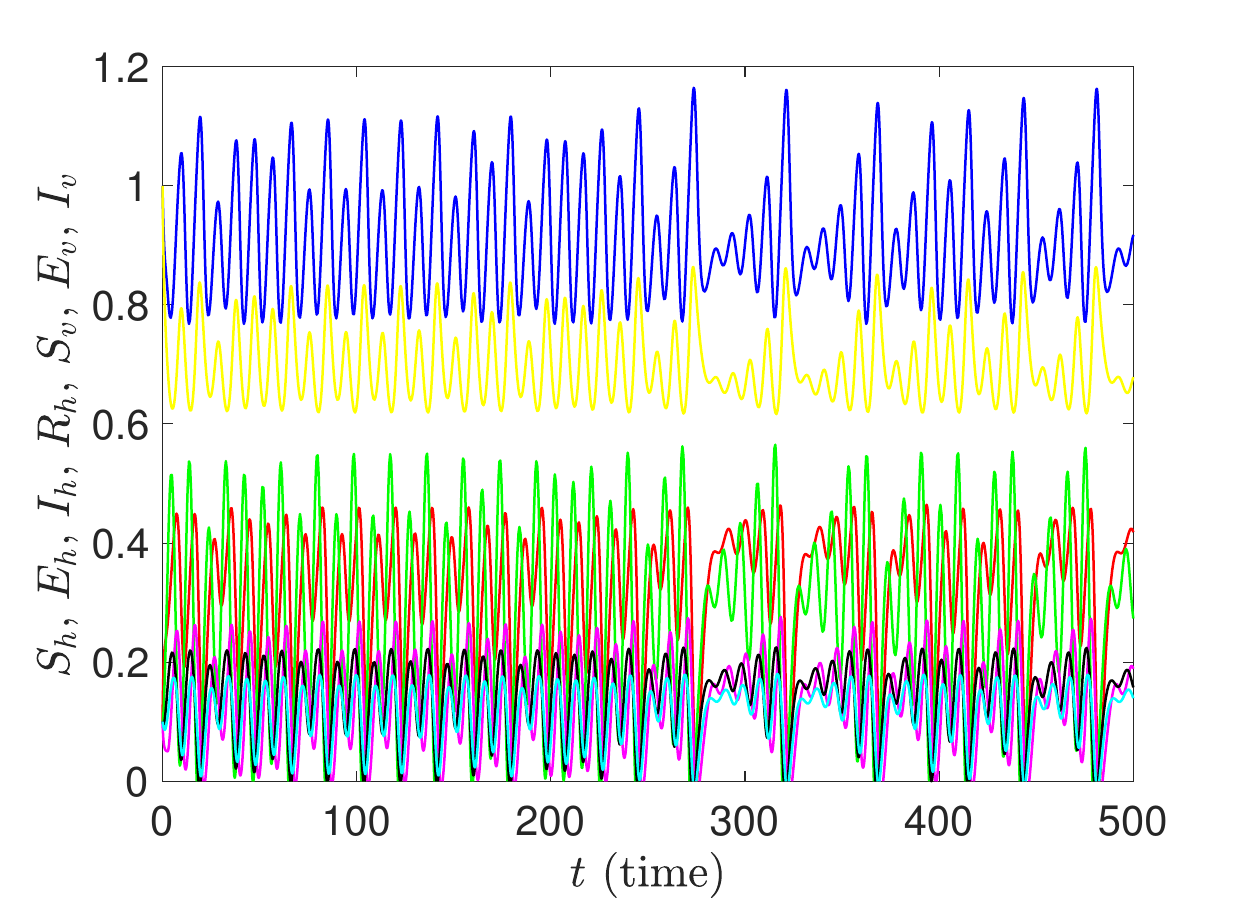} \\
\textbf{(d)}\\[6pt]
\end{tabular}
\caption{Time evolution of the system given by \eqref{equ21}-\eqref{equ27} in the presence of the delay in recovery of the individuals with   $\tau_h=\tau_v=0$ and $\tau_r=1.4$ (a), $\tau_r=1.7$ (b), $\tau_r=2.2$ (c) and $\tau_r=2.4$ (d). The initial conditions are set to $(1,0.1,0.1,0.1,1,0.1,0.1)$ and other parameters are given in the text. }
\label{Fig:TE2}
\end{figure}

\noindent In Fig.~\ref{Fig:Bif1}, stability of the susceptible human population along the equilibrium branches is shown with respect to  $\alpha_1$ and $b_{1h}$ for a non-delayed case. The number of unstable eigenvalues, denoted as $n_e$, is associated with blue ($n_e=0$) and red colors ($n_e=1$). Here only transcritical bifurcation where endemic equilibrium and disease free equilibrium intersect is detected. Bifurcation of the other components with respect to other parameters can be similarly performed, yet similar dynamics is obtained.

\begin{figure}[ht!]
\centering
\begin{tabular}{cccc}
\includegraphics[scale=0.28]{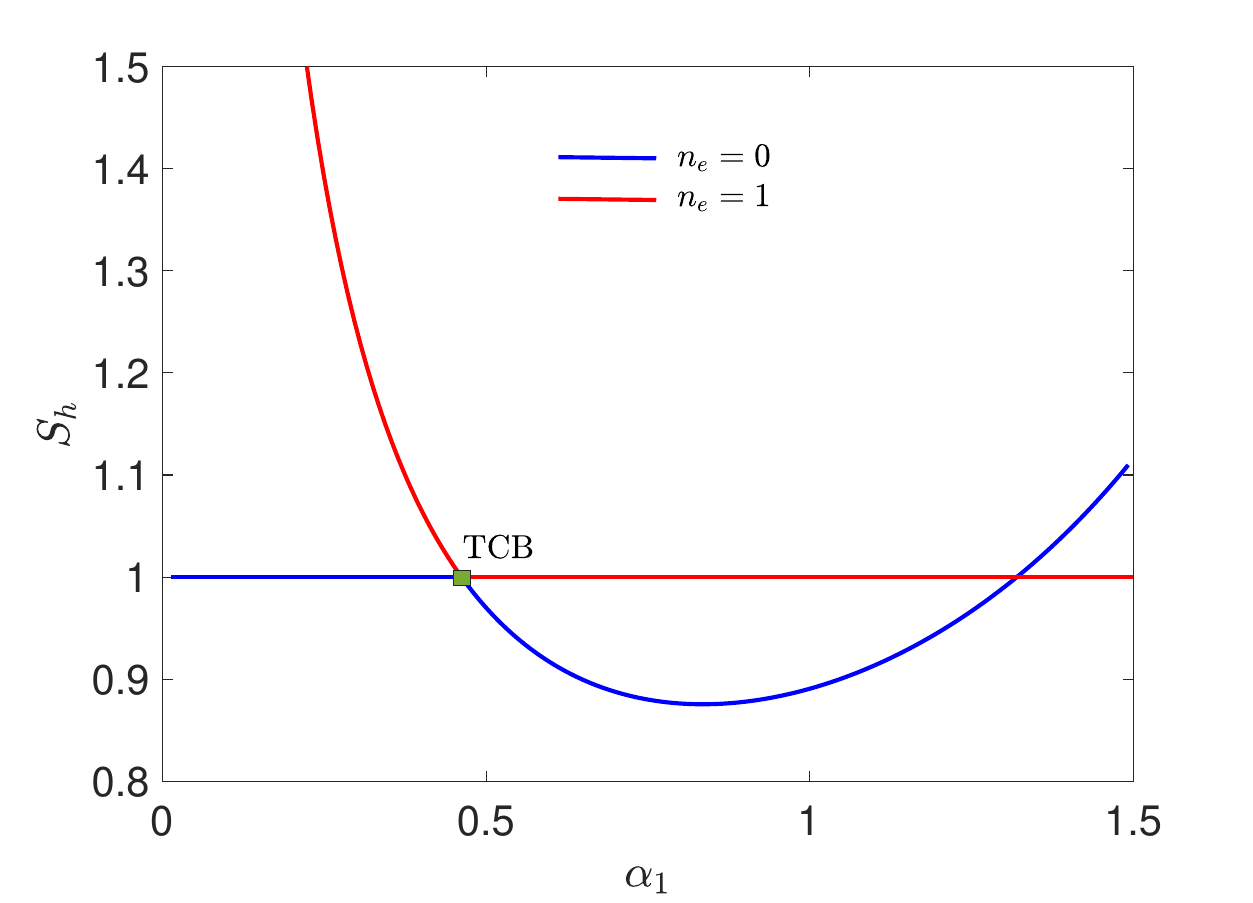} \\
\textbf{(a)}  \\[6pt]
\end{tabular}
\hspace{-1.05cm}
\vspace{-0.3cm}
\begin{tabular}{cccc}
\includegraphics[scale=0.28]{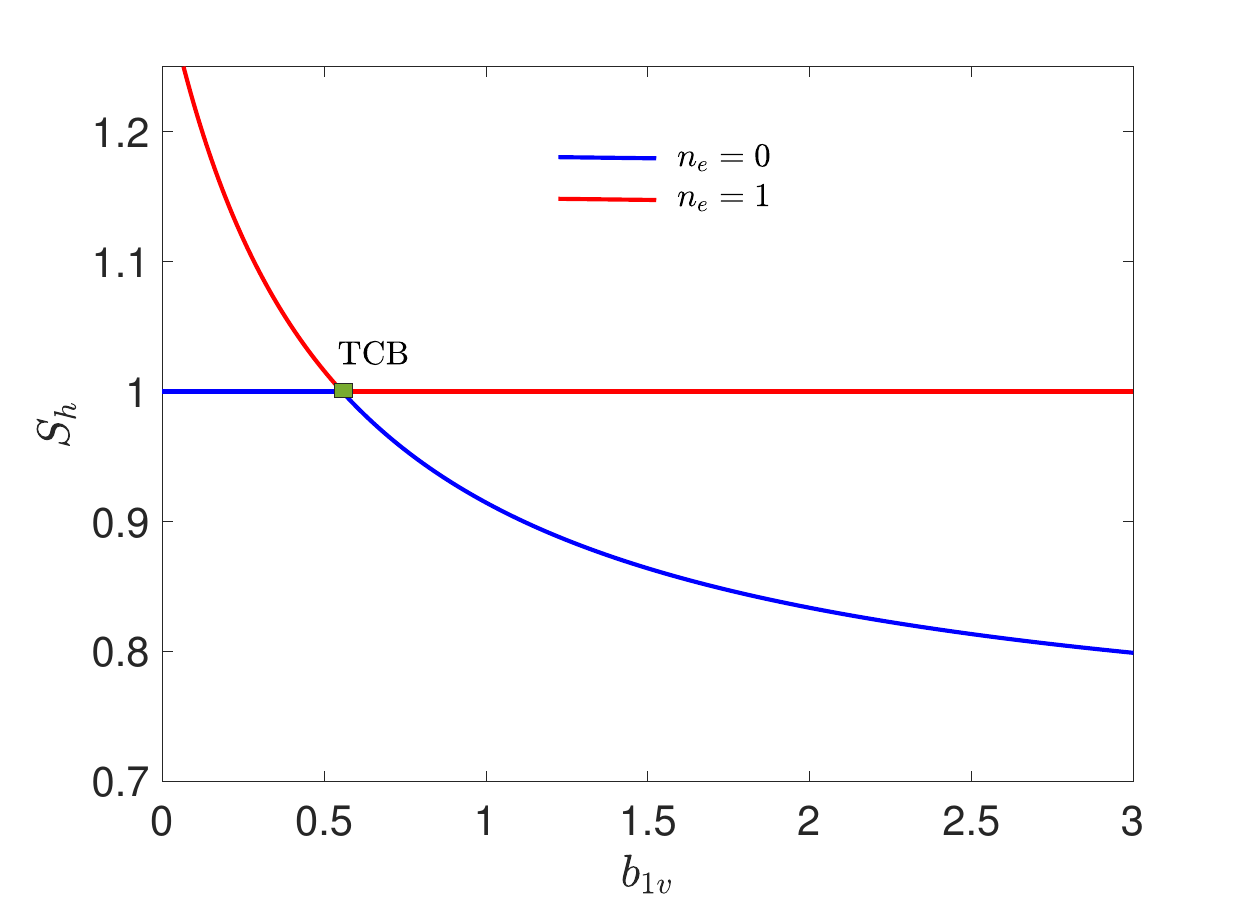} \\
\textbf{(b)}\\[6pt]
\end{tabular}
\caption{Bifurcation diagram of the variable for susceptible human population with respect to two parameters: $\alpha_1$ (a) and $\beta_{1h}$ (b) in the absence of delay constants, e.g.   $\tau_h=\tau_v=\tau_r=0$. The green square represents a transcritical bifurcation where endemic equilibrium and disease free equilibrium intersect. The initial conditions are set to $(1,0.1,0.1,0.1,1,0.1,0.1)$ and other parameters are given in the text. }
\label{Fig:Bif1}
\end{figure}

\subsection{Case I : $\tau_r= 2 \tau$}
%Now let $\lambda(\tau) = i \mu $ be a root of Eq.$ \eqref{Equ:Chactaui}$. Substituting this  into $\eqref{Equ:Chactaui}$, and considering the relation $e^{ −i x}$ = $\cos (x) - i \sin (x)$ gives the imaginary part as
Let $\lambda (\tau) = i \mu$ be a root of Eq.$ \eqref{Equ:Chactaui}$. Substituting this  into  $\eqref{Equ:Chactaui}$, and considering the relation $e^-{ ix}$ = $\cos (x) - i \sin (x)$ gives the imaginary part as
\begin{equation}
\begin{aligned}
M_1(\mu) \cos 2 \mu  \tau -M_2(\mu) \sin 2 \mu  \tau = M_3(\mu)
\end{aligned}
\label{Equ:Imag1i}
\end{equation}
and real part as
\begin{equation}
\begin{aligned}
M_2(\mu) \cos 2 \mu  \tau +M_1(\mu) \sin 2 \mu  \tau = M_4(\mu)
\end{aligned}
\label{Equ:real1i}
\end{equation}
where
\begin{align*}
M_1(\mu) =&\xi_h a_1 \mu^5+k_1\mu^3+k_2 \mu ,\quad
M_2(\mu) = -\xi_h \mu^6 +\xi_h a_2 \mu^4+k_3 \mu^2 +k_4 ,\\
M_3(\mu) =& \mu^7-a_2 \mu^5+a_4 \mu^3 -a_6 \mu ,\quad
M_4(\mu) = a_1 \mu^6 -a_3 \mu^4 +a_5 \mu^2.
\end{align*}
with $k_1=-xi_h a_3+d_1$, $k_2=\xi_h a_5 -d_3$, $k_3=-\xi_h a_4+d_2$ and $\xi_h a_6-d_4$ where
\begin{equation}
d_j = \begin{bmatrix}
m \alpha_1 \omega \\ \alpha_1 \alpha_2 f h
\end{bmatrix} \cdot \begin{bmatrix}
b_{j-1} \\ c_{j-1}
\end{bmatrix}, \quad j=1,2,3,4,
\end{equation}
for which $b_0=c_0=1$.
With a simple calculation, it can be shown from equations \eqref{Equ:Imag1i}-\eqref{Equ:real1i} that $M_1^2+M_2^2=M_3^2+M_4^2$, leading to a fourteenth   order equation:
\begin{equation}
\mu^{14} + B_1 \mu^{12} + B_2 \mu^{10} + B_3 \mu^8 + B_4 \mu^6+ B_5 \mu^4 + B_6 \mu^2 + B_7 = 0,
\label{Equ:mu12i}
\end{equation}
where
\begin{align*}
B_1= & a_1^2-2 a_2-\xi_h^2,\\
B_2 = &  - a_1^2 \xi_h^2 - 2 a_3 a_1 + a_2^2 + 2 a_2 \xi_h^2 + 2 a_4,\\
B_3 = &  2 a_1 a_5 - 2 a_6 - 2 a_2 a_4 + 2 k_3 \xi_h + a_3^2 - a_2^2 \xi_h^2 - 2 a_1 k_1 \xi_h,\\
B_4 = &  a_4^2 - k_1^2 + 2 a_2 a_6 - 2 a_3 a_5 + 2 k_4 \xi_h - 2 a_1 k_2 \xi_h - 2 a_2 k_3 \xi_h,\\
B_5 =& a_5^2 - k_3^2 - 2 a_4 a_6 - 2 k_1 k_2 - 2 a_2 k_4 \xi_h,\\
B_6 = & a_6^2 - k_2^2 - 2 k_3 k_4,\\
B_7 = &  -k_4^2.
\end{align*}
%%%
%%%
%%%
 Taking $v=\mu^2$, Eq.~\eqref{Equ:mu12i} equation can be rewritten as
\begin{equation}
v^7 + B_1 v^6 + B_2 v^5 + B_3 v^4 + B_4 v^3+ B_5 v^2 + B_6 v +B_7 = 0.
\label{Equ:v1i}
\end{equation}
The default parameter space given in Sec.~\ref{Sec:NonD} guarantee at least one positive root for Equ.~\eqref{Equ:v1i}. Assuming $\mu_c$ is the square root of positive root of \eqref{Equ:v1i} and using Eqs. \eqref{Equ:Imag1i}-\eqref{Equ:real1i} we find the critical threshold for time delay
\begin{equation}
\tau_{j}^c = \frac{1}{2 \mu_c} \cos^{-1}\left(\frac{M_1 M_3 + M_2 M_4}{M_1^2+M_2^2}\right)+\frac{ \pi j}{ \mu_c}, \quad j=1,2,\cdots,
\label{Equ:Critic}
\end{equation}
where $\mu_c$ is a positive root.  Differentiating both sides of Eq. \eqref{Equ:Psii1ii} with respect to $\tau$ yields
%\begin{equation}
%\begin{aligned}
%\left[\frac{\d \lambda }{\d \tau}\right]^{-1} =& -\frac{\Omega_E(\lambda,\tau,\tau, 2 \tau)+(\lambda+c_{1v}) (6 \lambda^5 + 5 a_1 \lambda^4 + 4 a_2 \lambda^3 + 3 a_3 \lambda^2 + 2 a_4 \lambda +a_5)}{2 b_0 \lambda (\lambda+c_{1v}) (\lambda^2+b_1 \lambda+b_2)\e^{-2 \lambda \tau} }\\
%& + \frac{2 \lambda +b_1}{2 \lambda (\lambda^2+b_1 \lambda+b_2)} -\frac{\tau}{\lambda},
%\end{aligned}
%\end{equation}
%Further, we have
\begin{equation}
\begin{aligned}
\left.\textrm{Re} \left[\frac{\d \lambda }{\d \tau}\right]\right|_{\tau=\tau_c} = \frac{\Psi_E '(\nu_c)}{M_1^2+M_2^2},  \qquad \nu_c= \mu_c^2.
\end{aligned}
\end{equation}
Thus if $\Psi_E(\mu_c^2) \neq 0$ holds, then $\textrm{Re}\left[\left(\frac{\d \lambda }{\d \tau}\right)^{-1}\right] \neq 0$ for the existence of  Hopf bifurcation theorem for a system with time delay \cite{hassard1981theory}.

\noindent In Fig.~\ref{Fig:Bif2}, the effect of time delays  is investigated using a numerical continuation and the stability analyses of both endemic and disease free states are determined under parameter variation. For example, we study the case where the dynamics is determined  parameters $\alpha_1$, $b_{1h}$, $b_{1v}$ and $c_{1v}$ are varied with $\tau_h=\tau_v=\tau=0.62$ and $\tau_r=1.24$. In this case, the system may exhibits Hopf bifurcation (purple square) in addition to transcritical bifurcation (green square). The blue line corresponds to stable state, the red dashed line corresponds to unstable  periodic orbits emanating from Hopf bifurcations with $(n_e=1)$, yellow  and gray lines stand for the unstable branches, for which the number of eigenvalues with positive real parts is respectively $n_e=2$ and $n_e=3$.
In all cases, instability occurs through Hopf points where unstable equilibrium is surrounded by unstable periodic orbits. Note that there many be countably infinitely many roots of the characteristic equation  Eq.~ \eqref{Equ:Mat} of the system \eqref{equ21}-\eqref{equ27} at the endemic equilibrium $\Sigma_1^*=(0.8906, 0.4102, 0.3418, 0.1709, 0.6642, 0.1740, 0.1450)$. However, it may be enough to consider the eigenvalues sufficiently close to imaginary axis. Therefore, the  minimal real parts of the eigenvalues are set to $\textrm{real}(\lambda)=-3$ for stability.

\begin{figure}[ht!]
\centering
\begin{tabular}{cccc}
\includegraphics[scale=0.28]{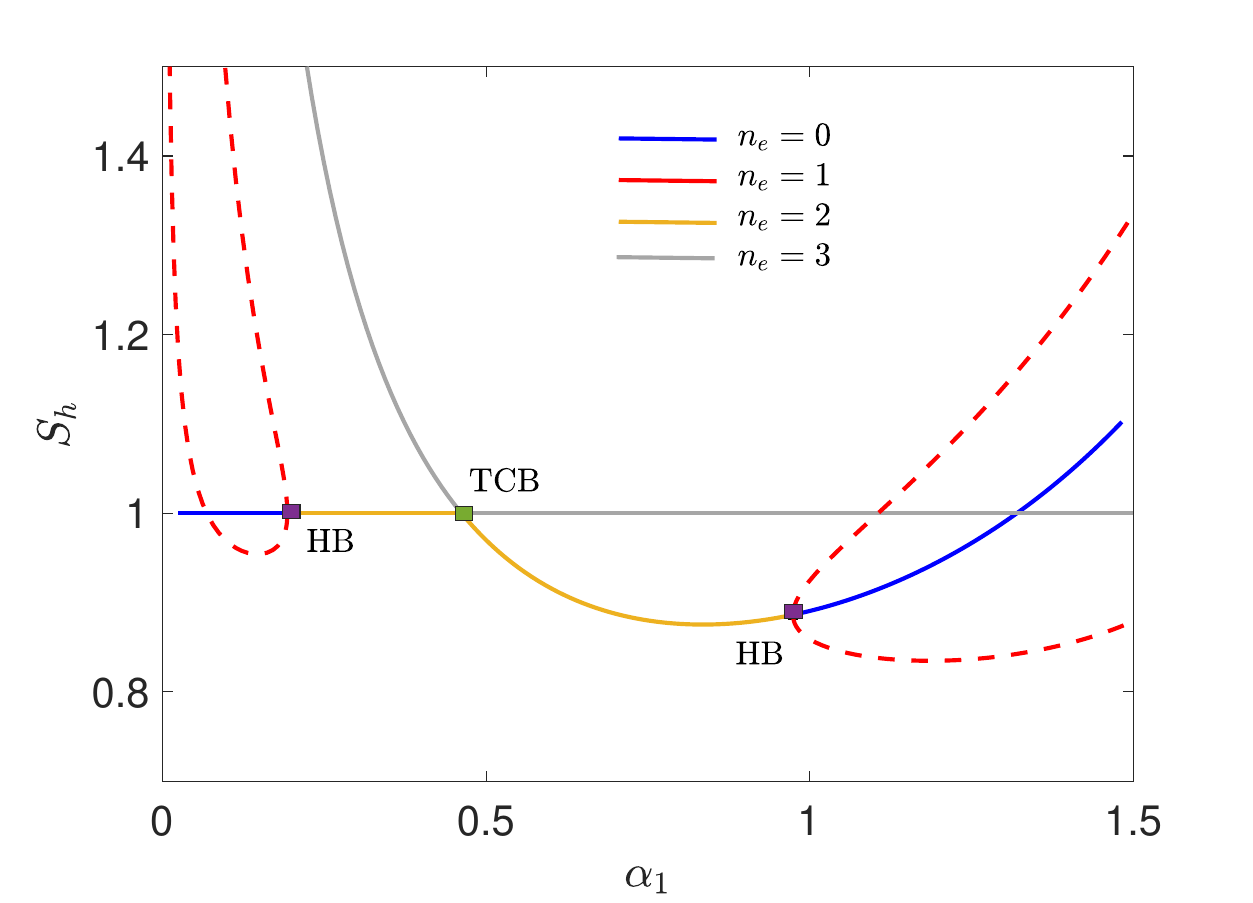} \\
\textbf{(a)}  \\[6pt]
\end{tabular}
\hspace{-1.05cm}
\vspace{-0.3cm}
\begin{tabular}{cccc}
\includegraphics[scale=0.28]{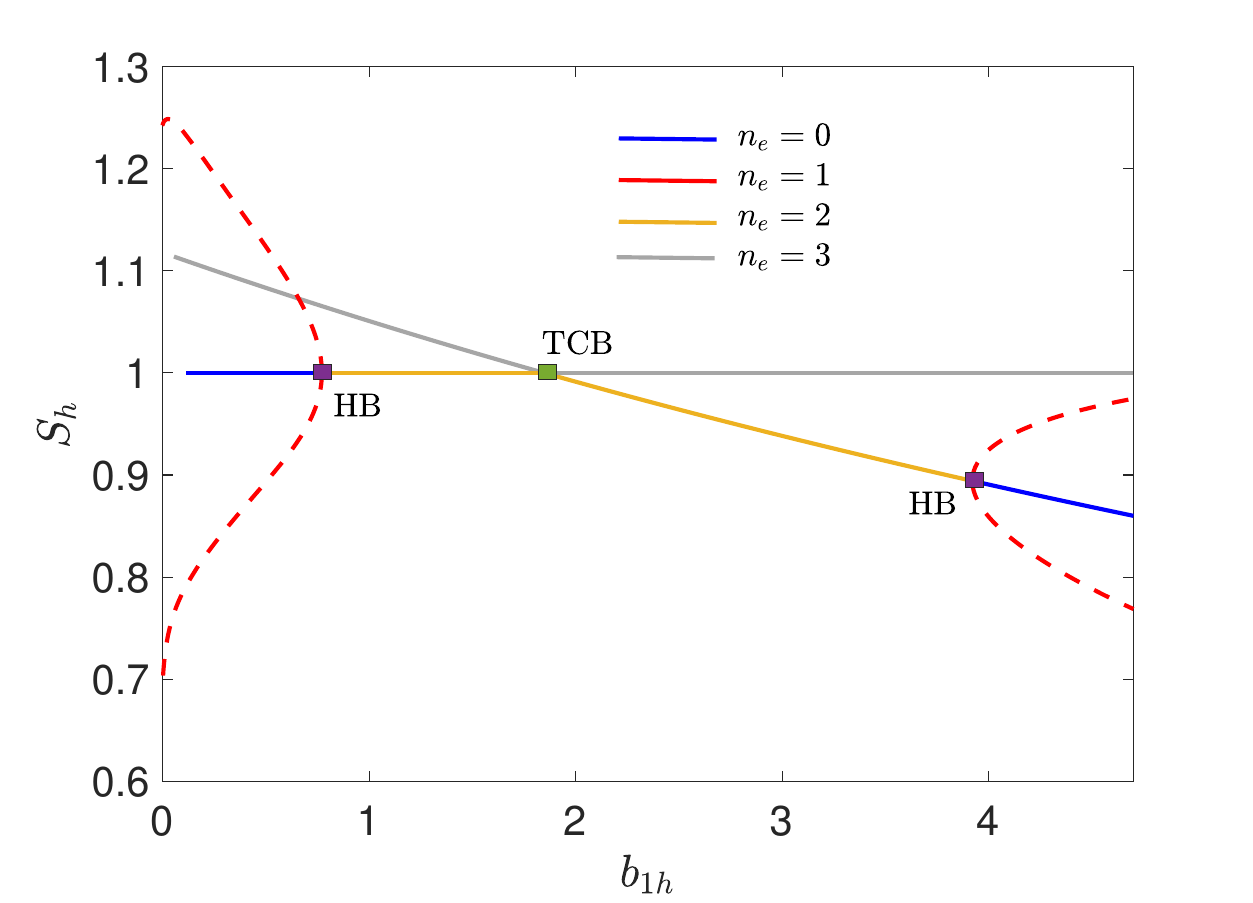} \\
\textbf{(b)}\\[6pt]
\end{tabular}
\begin{tabular}{cccc}
\includegraphics[scale=0.28]{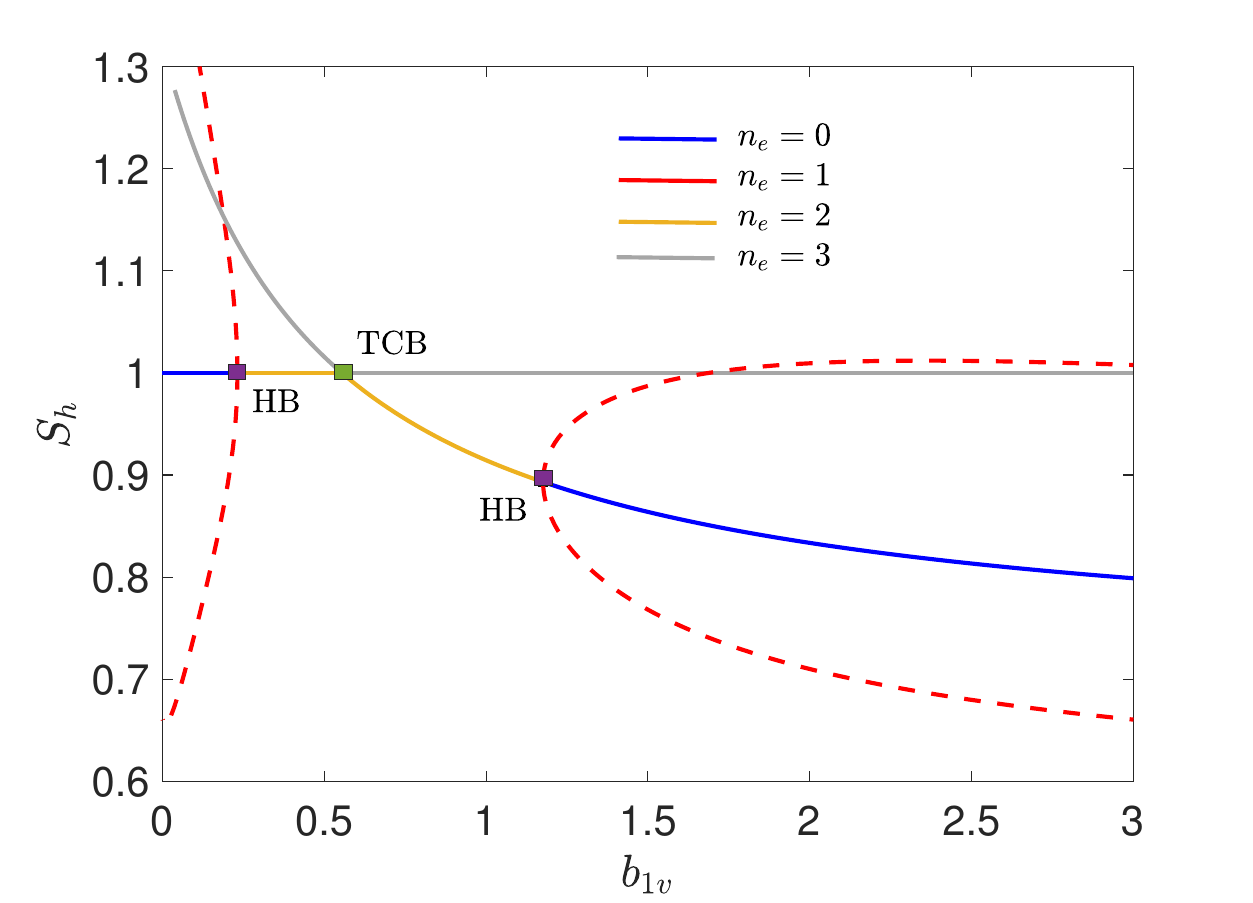} \\
\textbf{(c)}\\[6pt]
\end{tabular}
\hspace{-1.05cm}
\begin{tabular}{cccc}
\includegraphics[scale=0.28]{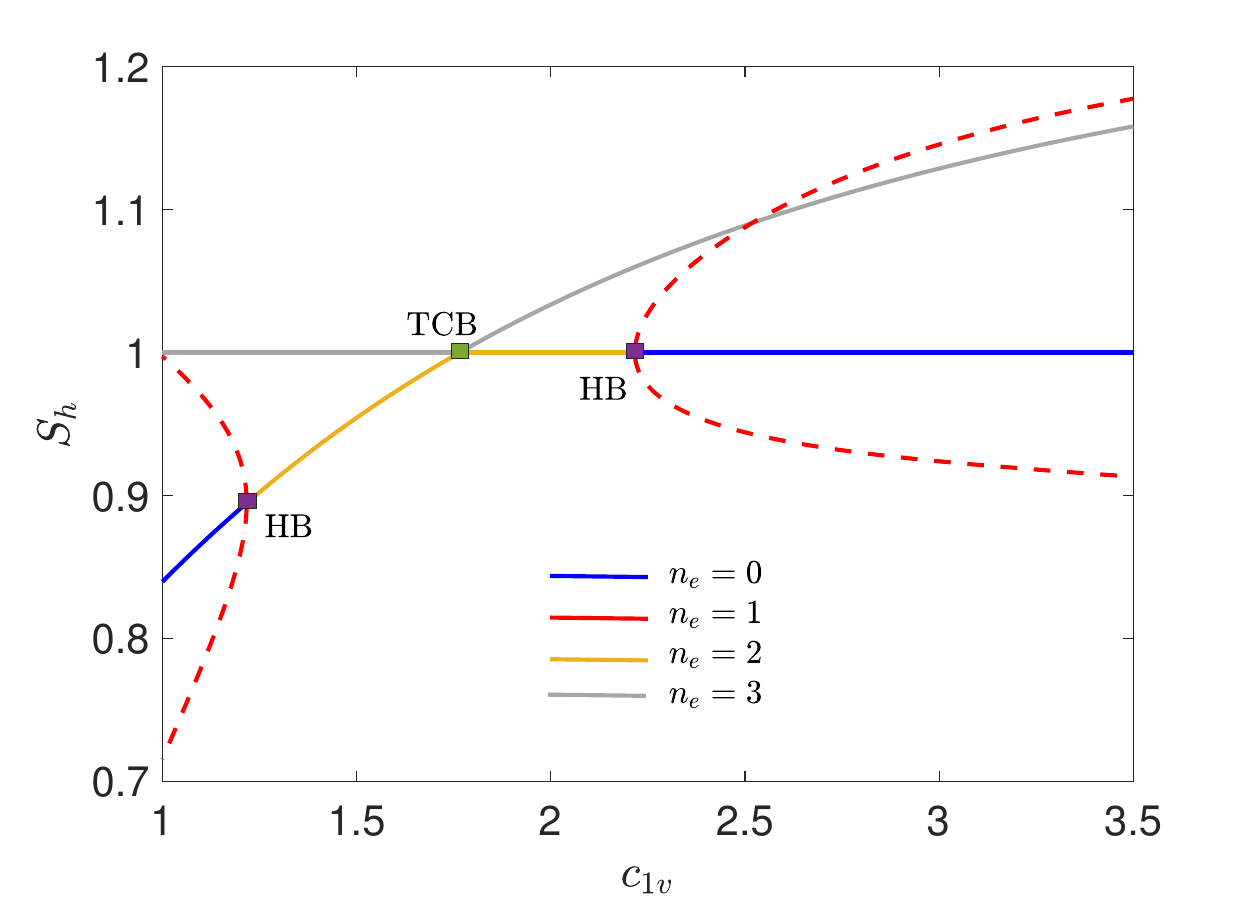} \\
\textbf{(d)}\\[6pt]
\end{tabular}
\caption{Single parameter numerical continuation of susceptiple human population with respect to parameters $\alpha_1$ (a), $b_{1h}$ (b), $b_{1v}$ (c) and $c_{1v}$ (d) with $\tau_h=\tau_v=\tau=0.62$ and $\tau_r=1.24$. The colored lines represent the number of  eigenvalues with positive real parts. The dashed lines stand for the branches emanating from Hopf bifurcation and represent the maximum of the periodic orbits. Purple and green squares respectively represent Hopf and transcritical bifurcations. }
\label{Fig:Bif2}
\end{figure}

\subsection{Case II : $\tau_r \neq 0$ and $\tau=0$}

Now it is assumed that there is no average extrinsic and intrinsic incubation time, e.g.
$\tau_h=\tau_v=\tau=0 $ and   the infectious people start  recovering after a period of time $\tau_r$. Hence the characteristic equation given by \eqref{Equ:Chactaui} becomes
\begin{equation*}
\begin{aligned}
&\lambda \left[\lambda^6+a_1 \lambda^5 + a_2 \lambda^4 +a_3 \lambda^3+a_4 \lambda^2+a_5 \lambda +a_6\right]-\xi_h \e^{-\lambda \tau_r} \left[\lambda^6+a_1 \lambda^5 \right. \\& + \left. a_2 \lambda^4 +a_3 \lambda^3+a_4 \lambda^2 +a_5 \lambda + a_6\right]-m \alpha_1 \omega \e^{-\lambda \tau_r} (\lambda^3+b_1 \lambda^2+b_2 \lambda +b_3)\\&
-\alpha_1 \alpha_2 f h (\lambda^3+c_1 \lambda^2 +c_2 \lambda +c_3)=0.
\end{aligned}
\end{equation*}
Similarly, substituting $\lambda=i \mu $ leads to coefficients for real and imaginary parts given in \eqref{Equ:Imag1i} and \eqref{Equ:real1i} as
\begin{align*}
\hat{M}_1(\mu) &=\xi_h a_1 \mu^5+k_1\mu^3+k_2 \mu,\\
\hat{M}_2(\mu) &= -\xi_h \mu^6 +\xi_h a_2 \mu^4+k_3 \mu^2 +k_4 ,\\
\hat{M}_3(\mu) &= \mu^7-a_2 \mu^5+k_5 \mu^3 + k_6 \mu ,\\
\hat{M}_4 (\mu) &= a_1 \mu^6 -a_3 \mu^4 +k_7 \mu^2 +k_8,
\end{align*}
with $k_1=-\xi_h a_3+ m \alpha_1 \omega$, $k_2=\xi_h a_5 -m \alpha_1 \omega b_2$, $k_3=-\xi_h a_4+m \alpha_1 \omega b_1$ and $k_4=\xi_h a_6-m \alpha_1 \omega b_3$, $k_5=a_4-\alpha_1 \alpha_2 f h$, $k_6=-a_6+\alpha_1 \alpha_2 f h c_2$, $k_7=a_5-\alpha_1 \alpha_2 f h c_1$, $k_8=\alpha_1 \alpha_2 f h c_3$. Besides, $\hat{M}_1^2+\hat{M}_2^2=\hat{M}_3^2+\hat{M}_4^2$ leads to a fourteenth order equation
\begin{equation}
\mu^{14} + \hat{B}_1 \mu^{12} + \hat{B}_2 \mu^{10} + \hat{B}_3 \mu^8 + \hat{B}_4 \mu^6+ \hat{B}_5 \mu^4 + \hat{B}_6 \mu^2 + \hat{B}_7  = 0,
\label{Equ:mu12ia}
\end{equation}
where
\begin{align*}
\hat{B}_1= &A_1,\\
\hat{B}_2 = &- a_1^2 \xi_h^2 - 2 a_3 a_1 + a_2^2 + 2 a_2 \xi_h^2 + 2*k_5,\\
\hat{B}_3 = &2*k_6 - 2 a_2 k_5 + 2 a_1 k_7 + 2 k_3 \xi_h + a_3^2 - a_2^2 \xi_h^2 - 2 a_1 k_1 \xi_h,\\
\hat{B}_4 = &- k_1^2 + k_5^2 - 2 a_2 k_6 + 2 a_1 k_8 - 2 a_3 k_7 + 2 k_4 \xi_h - 2 a_1 k_2 \xi_h - 2 a_2 k_3 \xi_h,\\
\hat{B}_5 =&  - k_3^2 + k_7^2 - 2 a_3 k_8 - 2 k_1 k_2 + 2 k_5 k_6 - 2 a_2 k_4 \xi_h,\\
\hat{B}_6 = &  - k_2^2 + k_6^2 - 2 k_3 k_4 + 2 k_7 k_8,\\
\hat{B}_7 = & k_8^2 - k_4^2.
\end{align*}
Following the ideas presented for Case I, we find the critical threshold for time delay as
\begin{equation}
\tau_{r_j}^c = \frac{1}{ \mu_c} \cos^{-1}\left(\frac{\hat{M}_1 \hat{M}_3 + \hat{M}_2 \hat{M}_4}{\hat{M}_1^2+\hat{M}_1^2}\right)+\frac{2 \pi j}{ \mu_c}, \quad j=1,2,\cdots,
\label{Equ:Critici}
\end{equation}
where $\mu_c$ is a positive root.

\noindent In Fig.~\ref{Fig:Bif3}, stability of the susceptible human population is shown with respect to the variations of parameters $\alpha_1$ (a), $b_{1h}$ (b), $b_{1v}$ (c) and $c_{1v}$ (d) in the absence of intrinsic and extrinsic incubation times ($\tau_h=\tau_v=0$) and in the presence of delay in the recovery period ($\tau_r=2.4$). As seen, period doubling bifurcations (cyan square), where stability of the dynamics is switched from stable to unstable,  arise in the periodic orbits emanating from Hopf bifurcations. Besides, interestingly the unstable branches emanating from Hopf point for disease free equilibrium exhibits spiral like trajectories.

\begin{figure}[ht!]
\centering
\begin{tabular}{cccc}
\includegraphics[scale=0.28]{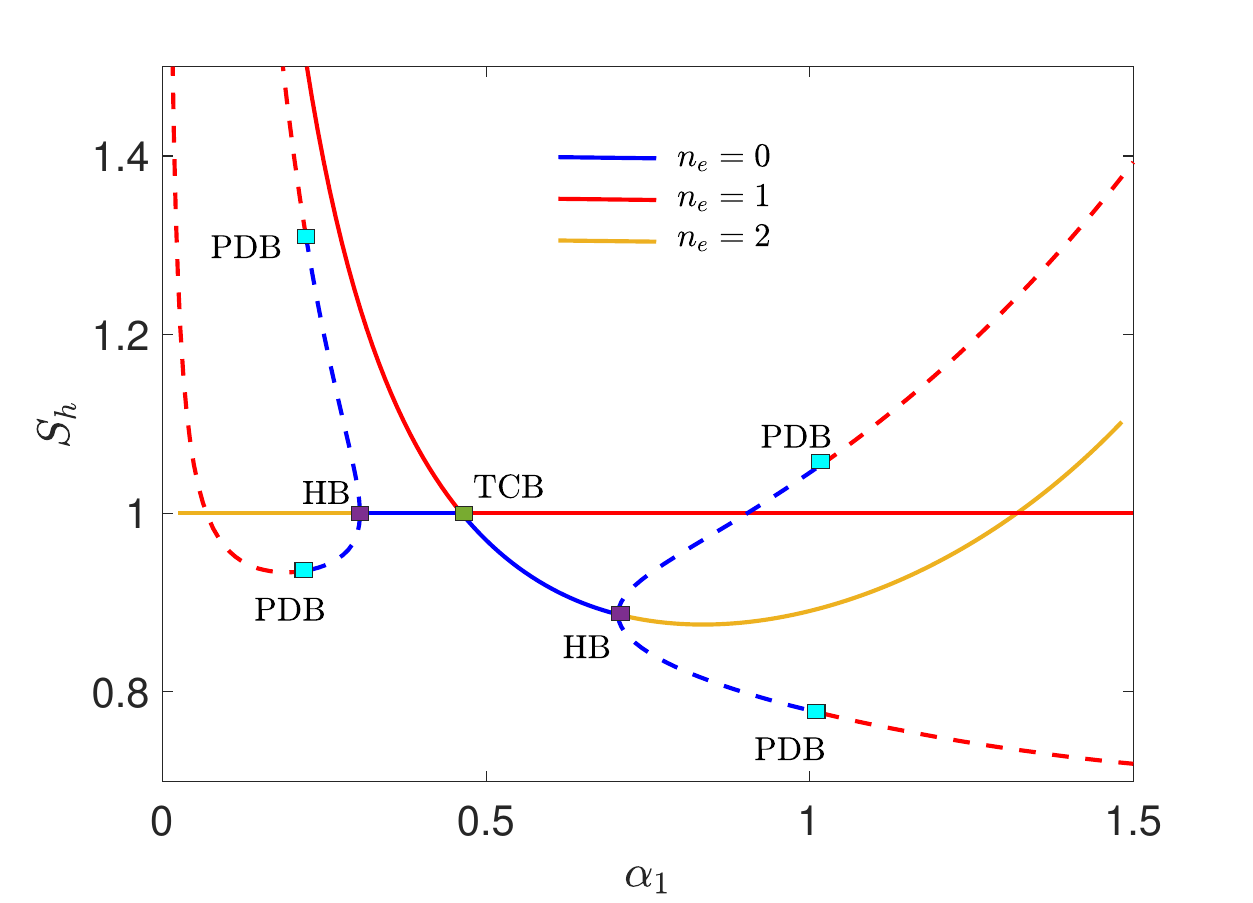} \\
\textbf{(a)}  \\[6pt]
\end{tabular}
\hspace{-1.05cm}
\vspace{-0.3cm}
\begin{tabular}{cccc}
\includegraphics[scale=0.28]{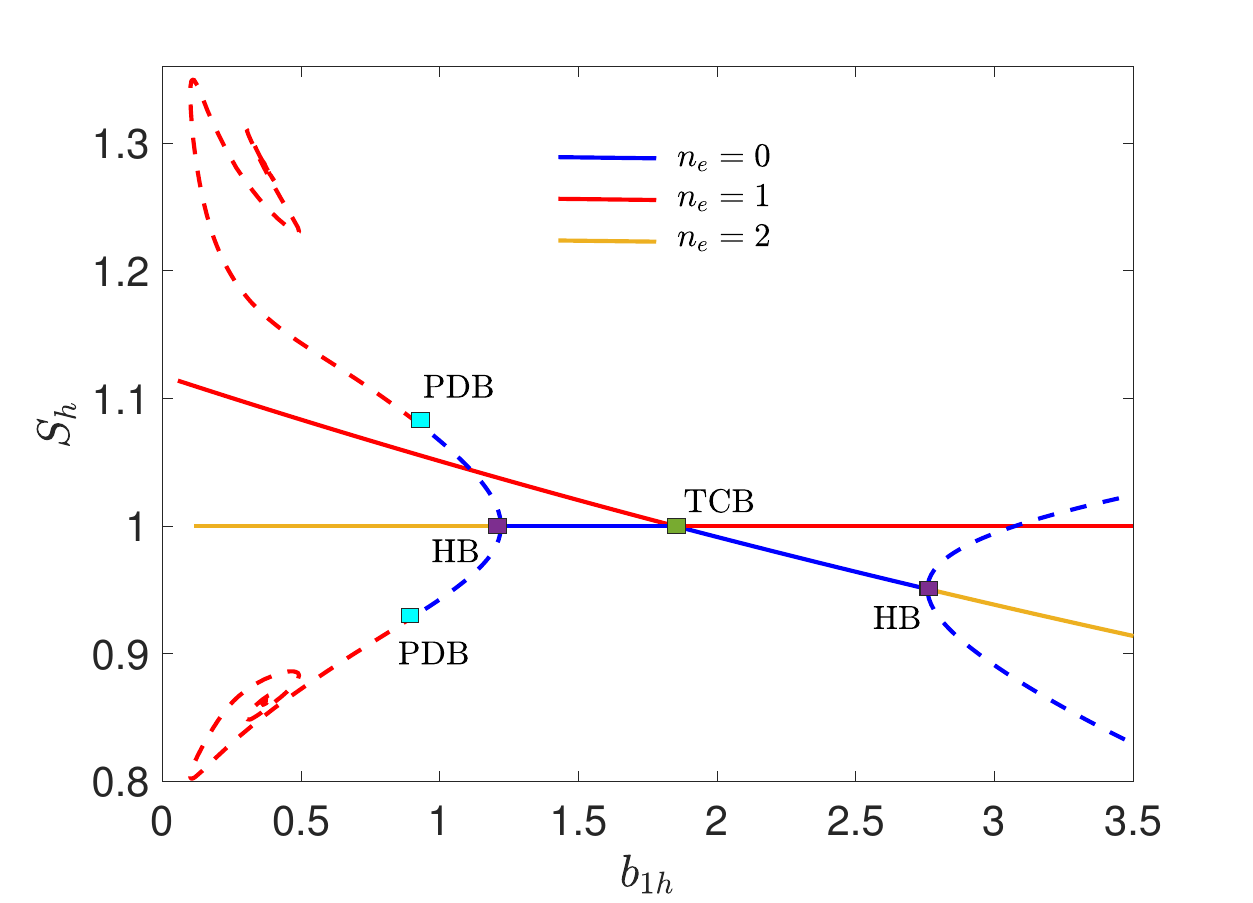} \\
\textbf{(b)}\\[6pt]
\end{tabular}
\begin{tabular}{cccc}
\includegraphics[scale=0.28]{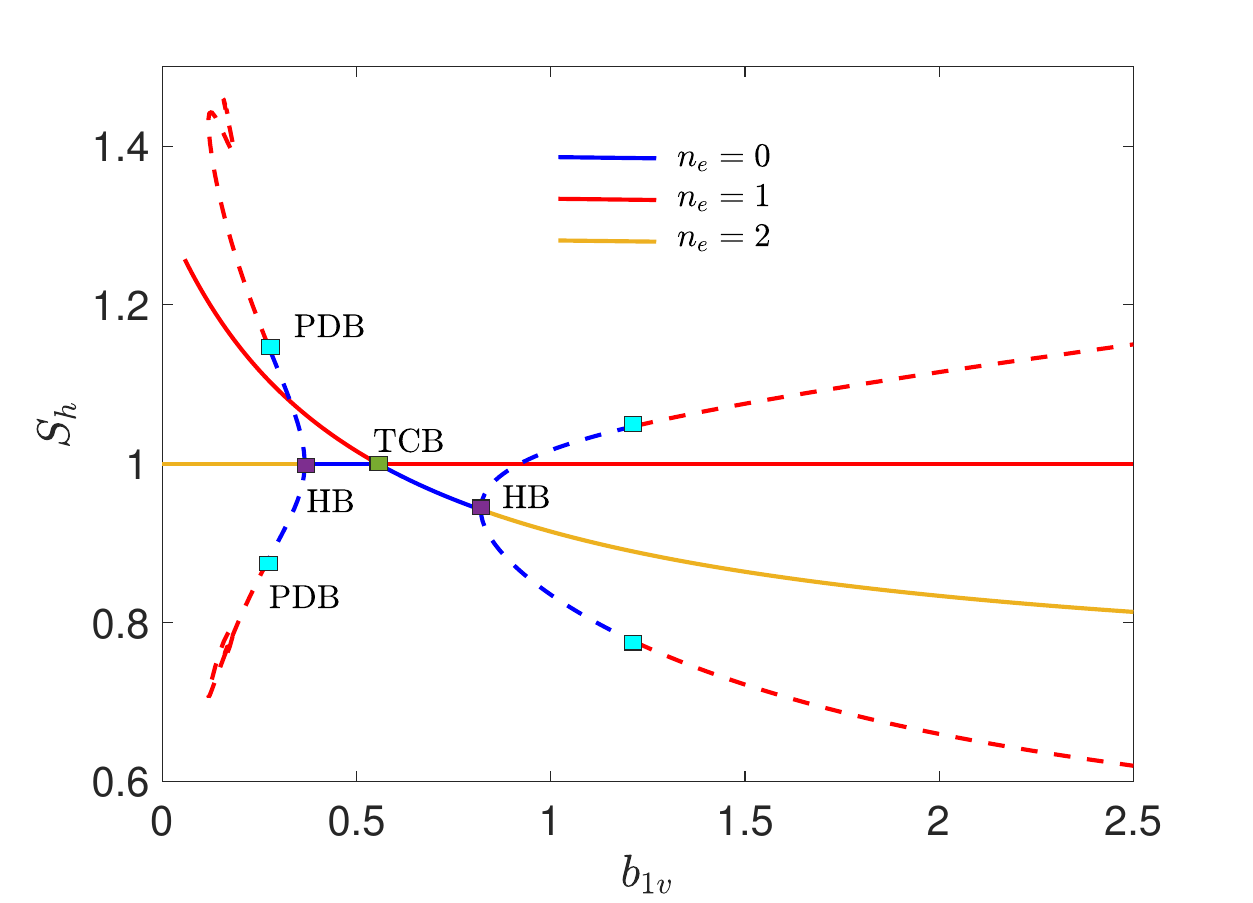} \\
\textbf{(c)}\\[6pt]
\end{tabular}
\hspace{-1.05cm}
\begin{tabular}{cccc}
\includegraphics[scale=0.28]{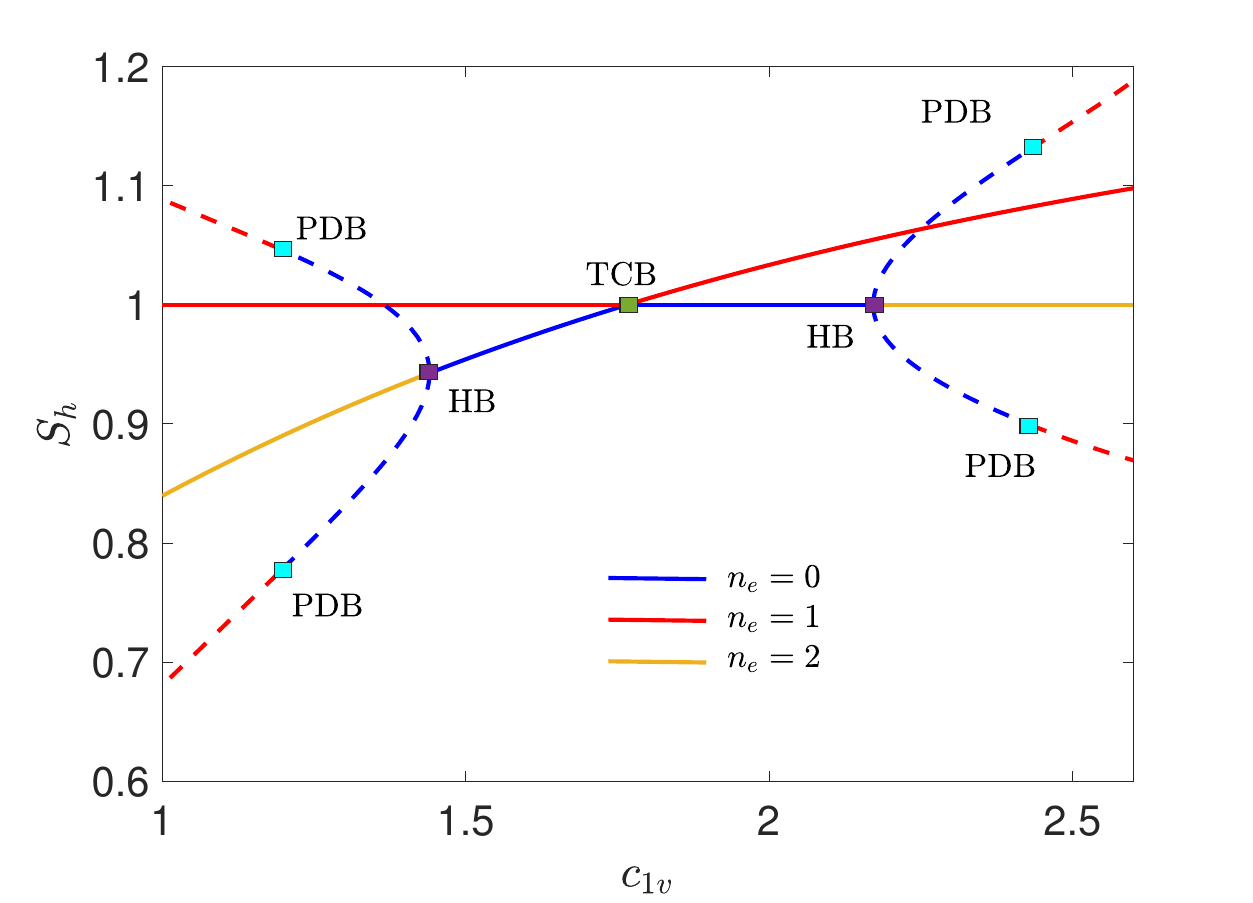} \\
\textbf{(d)}\\[6pt]
\end{tabular}
\caption{Single parameter numerical continuation of susceptible human population with respect to parameters $\alpha_1$ (a), $b_{1h}$ (b), $b_{1v}$ (c) and $c_{1v}$ (d) with $\tau_h=\tau_v=\tau=0$ and $\tau_r=2.4$. The colored lines represent the number of  eigenvalues with positive real parts. The dashed lines stand for the branches emanating from Hopf bifurcation and represent the maximum of the periodic orbits. Purple, cyan and green squares respectively represent Hopf, Period doubling and transcritical bifurcations. }
\label{Fig:Bif3}
\end{figure}

\subsection{Disease free equilibrium (DFE)}
Disease free equilibrium is found as $\Sigma_2^* = \left(1,0,0,0,\frac{1}{c_{1v}},0,0\right)$. Using the linearization argument the explicit form of the matrix given in Eq. \eqref{Equ:Mat} is given as follows
\[
J_\tau=\begin{pmatrix}
-1&0 & 0 & \omega & 0 & 0 &- f\e^{-\lambda \tau_h}\\
0& -\sigma_h& 0 &0&0&0&\alpha_1 f \e^{-\lambda \tau_h}\\
0 & 1 & -\xi_h \e^{-\lambda \tau_r} & 0&0&0&0\\
0&0& \e^{-\lambda \tau_r}& -\eta_h&0 &0&0\\
0&0&-h\e^{-\lambda \tau_v}& 0 & -c_{1v}&0&0\\
0&0& \alpha_2 h\e^{-\lambda \tau_v}&0&0&-c_{2v}&0\\
0&0&0&0&0&1&-c_{1v}
\end{pmatrix}.
\]
The calculations performed in Sec.~\ref{Sec:EndEqu}   are valid here with the exception of Eq. \eqref{Equ:b1hhi}. Since only susceptible human and susceptible vector population exist for disease free state, the coefficients in Eq. \eqref{Equ:b1hhi} becomes
\begin{equation}
d=0, \quad  f=b_{1h}, \quad h=\frac{b_{1v}}{c_{1v}}, \quad g=0.
\label{Equ:b1h0i}
\end{equation}
As an example let us concentrate on the case where $\tau_r \neq 0$ and $\tau_h=\tau_v=0$, leading to
\begin{equation}
\begin{aligned}
\Psi_{\Sigma_2^*}(\lambda) =  (\lambda+1)(\lambda+\eta_h)(\lambda+c_{1v}) \Omega_{\Sigma_2^*}(\lambda,0,0, \tau_r)=0 ,
\end{aligned}
\label{Equ:DFi}
\end{equation}
where
\begin{equation}
\Omega_{\Sigma_2^*}(\lambda,0,0, \tau_r)= (\lambda+c_{1v})(\lambda+c_{2v})(\lambda+\sigma_h)(\lambda+\xi_h \e^{-\lambda \tau_r})
- f h \alpha_1 \alpha_2 f h.
\label{Equ:DF2i}
\end{equation}
It is straightforward to see from Eq. \eqref{Equ:DFi} that three negative real eigenvalues arising from the disease-free system are $\lambda_1=-c_{1v}$, $\lambda_2=-\eta_h$ and $\lambda_3=-1$, and the other four eigenvalues can be determined depending on the delay parameter $\tau_r$, for which  Eq. \eqref{Equ:DF2i} can be rewritten as
\begin{equation}
\Omega_{\Sigma_2^*}(\lambda,0,0, \tau_r)=( \lambda^3 + d_1 \lambda^2 + d_2 \lambda + d_3 ) ( \lambda  + \xi \e^{- \lambda \tau_r}) -f h \alpha_1 \alpha_2=0,
\label{Equ:DF3i}
\end{equation}
where
\begin{align*}
d_1&= c_{1v}+c_{2v}+\sigma_h ,\\
d_2& = c_{1v} c_{2v} +c_{1v}\sigma_h + c_{2v} \sigma_h ,\\
d_3 &= c_{1v} c_{2v} \sigma_h.
\end{align*}
Let $i \mu_0$ is a solution of Eq. \eqref{Equ:DF3i}. Thus an eight order equation can be written as
\begin{equation}
\begin{aligned}
& \mu_0^8 +(d_1^2-2 d_2-\xi^2) \mu^6 + (-2 d_1 d_3+d_2^2+2\xi^2 d_2-\xi^2 d_1^2-2 f h \alpha_1 \alpha_2) \mu^4 \\& +(d_3^2-\xi^2 d_2^2+2 \xi^2 d_1 d_3+2 d_2 f h \alpha_1 \alpha_2)\mu^2+(f h \alpha_1 \alpha_2)^2-d_3^2 \xi^2=0.
\end{aligned}
\end{equation}
 Setting $\nu_0= \mu_0^2$, a fourth order polynomial is found as in the following
\begin{equation}
f(\nu_0) = \nu_0^4 + D_1 \nu_0^3 + D_2 \nu_0^2 + D_3 \nu_0 + D_4=0,
\label{Equ:fnui}
\end{equation}
with $D_1 = d_1^2-2 d_2-\xi^2$, $D_2 = -2 d_1 d_3+d_2^2+2\xi^2 d_2-\xi^2 d_1^2-2 f h \alpha_1 \alpha_2$, $D_3 = d_3^2-\xi^2 d_2^2+2 \xi^2 d_1 d_3+2 d_2 f h \alpha_1 \alpha_2$ and $D_4= (f h \alpha_1 \alpha_2)^2-d_3^2 \xi^2$, and therefore
\begin{equation}
f'(\nu_0) = 4 \nu_0^3 + 3 D_1 \nu_0^2 + 2 D_2 \nu_0 +D_3 =0.
\label{Equ:Fpii}
\end{equation}
Denoting $y_D = \nu_0+ \frac{3 D_1}{4}$, Eq. \eqref{Equ:Fpii} becomes
\begin{equation}
y_0^3 + e_1 y_0 + e_2=0,
\end{equation}
where
\[e_1=\frac{8 D_2- 3 D_1^2}{16} \quad \textrm{and} \quad e_2 = \frac{D_1^3-4 D_2 D_1+ 32 D_3}{32}.\]
The  solutions of Eq. \eqref{Equ:fnui} can be obtained based on the argument given in \cite{li2005zeros}:
\begin{align*}
{y_0}_1 &= \sqrt[3]{-\frac{e_2}{2}+\sqrt{\hat{\Delta}}} + \sqrt[3]{-\frac{e_2}{2}-\sqrt{\hat{\Delta}}},\\
{y_0}_2 &= \zeta \sqrt[3]{-\frac{e_2}{2}+\sqrt{\hat{\Delta}}} + \zeta^2 \sqrt[3]{-\frac{e_2}{2}-\sqrt{\hat{\Delta}}},\\
{y_0}_3 &= \zeta^2 \sqrt[3]{-\frac{e_2}{2}+\sqrt{\hat{\Delta}}} + \zeta \sqrt[3]{-\frac{e_2}{2}-\sqrt{\hat{\Delta}}},\\
\end{align*}
where \[ \hat{\Delta} = \left(\frac{e_1}{3}\right)^3 + \left(\frac{e_2}{2}\right)^2,\qquad \zeta=\frac{-1+\sqrt{3}i }{2},\quad \textrm{and}\quad {\nu_0}_j = {y_0}_j - \frac{3 D_1}{4}, \quad j=1,2,3.\]
\begin{lemma}
One has the following results for Eq. \eqref{Equ:fnui}:
\\
$\boldsymbol{C_1:}$ If $f h \alpha_1 \alpha_2 < d_3 \xi$ $(D_4<0)$  Eq. \eqref{Equ:fnui} has at least one positive root,\\
$\boldsymbol{C_2 :}$ If $f h \alpha_1 \alpha_2 \geq d_3 \xi$ and $\hat{\Delta}\geq 0$ then Eq. \eqref{Equ:fnui} has  positive roots if and only if ${\nu_0}_1>0 $ and $f({\nu_0}_1)<0$,\\
$\boldsymbol{C_3 :}$  If $f h \alpha_1 \alpha_2 \geq  d_3 \xi$ and $\hat{\Delta}< 0$ then Eq. \eqref{Equ:fnui} has  positive roots if and only if there exists at least one ${\nu_0}^*>0$ for which $f({\nu_0}^*)<0$. Here ${\nu_0}^* \in \{{\nu_0}_1, {\nu_0}_2, {\nu_0}_3 \}.$
\end{lemma}
If the condition $\boldsymbol{C_1}$ holds,  Eq. \eqref{Equ:fnui} has positive root $\nu_c$ such that Eq. \eqref{Equ:DF3i} has a pair of imaginary roots: $\pm i \mu_c = \pm i \sqrt{\nu_c}$. Besides using the similar argument given  in Eq. \eqref{Equ:Critici} where the components in Eq. \eqref{Equ:aii} are computed for $d=0$, $g=0$, $f=b_{1h}$ and $h=\frac{b_{1v}}{c_{1v}}$.
Furthermore, differentiating both side of Eq. \eqref{Equ:DFi} with respect to $\lambda$, it can be similarly shown that $\textrm{Re} \left[\frac{d \lambda}{d \tau}\right]^{-1} \neq 0.$

\begin{lemma}
When $d_1^2=2 d_2+\xi^2$ and $d_2^2 \xi^2=d_3^2+2 d_1 d_3 \xi^2 + 2 d_2 f h \alpha_1 \alpha_2$. Eq. \eqref{Equ:fnui} becomes
\begin{equation}
\mu_0^8+D_2  \mu_0^4+ D_4=0.
\label{Equ:Rfnui}
\end{equation}
Taking $x_D=v_0^2$ and $\Delta_{D} = D_2^2-4 D_4 $
 then we have the following
conclusions:
\begin{itemize}
\item When $D_2^2>4 D_4$, then Eq. \eqref{Equ:Rfnui} has two unequal real roots: $x_D^{(1,2)} = \frac{D_2\pm \sqrt{\Delta_D}}{2}$ then Eq. \eqref{Equ:fnui} has eight unequal real roots which are $\lambda_{1,2,3,4} =\pm \sqrt{x_D^{(1)}} $ and $\lambda_{5,6,7,8} =\pm \sqrt{x_D^{(2)}} $.
\item When $\Delta_D=0$, Eq. \eqref{Equ:Rfnui} has two equal real roots and when $\Delta_D<0$,  Eq. \eqref{Equ:Rfnui} has two unequal and conjugate imaginary roots $x_D^{(1,2)} = \frac{D_2\pm i \sqrt{\Delta_D}}{2}$ then Eq.\eqref{Equ:fnui} has eight unequal and conjugate imaginary roots.
\end{itemize}
\end{lemma}

\section{Conclusion}
\label{Sec:Conc}

In this research, a theoretical dengue fever propagation analysis for the extrinsic and intrinsic incubation period is modeled via a system of delay differential equations. As our focus is on the dynamics, we perform non-dimensionalization for model equation \eqref{Equ:11_17}. Here, the dengue fever non-dimensionalized model exhibits existence of dual equilibrium points: the Dengue-free equilibrium and the Dengue-endemic equilibrium. The Dengue-free equilibrium in the absence of delay is locally and globally stable at $R_0\leq1$. Otherwise, the Dengue-endemic equilibrium exists for $R_0>1$. In this context, $R_0$ represents a threshold that determines whether the dengue disease is eradicated or persists in the population of interest.

\noindent Since time delays appear in infectious groups $I_v$ and $I_h$, the dengue disease model presented in system \eqref{Equ:11_17} relies on the system of solutions at earlier periods. In fact, like in many real-life applications, time lags occur in many disease models and thus give much more reasonable insights into population dynamics. As illustrated in the time simulations given in Section \ref{Sec:Stability}, compared to the cases where there is no delay, a substantial change in the dynamics is observed, particularly with the increase in delay term $\tau_r$, representing the time after that infectious individuals begin to recover. Then the results are extended to include single-parameter bifurcation diagrams of the system and the effect of delay parameters on the system dynamics is explicitly presented. In this context, only transcritical bifurcation (TCB) has been detected for the case without delays. However, more complicated dynamics with period-doubling bifurcation and Hopf bifurcation have been observed in the occurrence of time lags. In fact, several authors have performed complex numerical experiments for bifurcation analysis for various infectious disease models \cite{kooi2013bifurcation,stollenwerk2022seasonally,zhang2022bifurcations}. This may also play an important role in understanding different epidemic diseases. For instance, by studying period-doubling bifurcations, we can potentially identify parameter regions where complex dynamics, such as chaos, arise \cite{stollenwerk2022seasonally}.

\noindent The numerical simulations in this paper have been carried out via MATLAB's dde23 solver. The results for numerical continuation have been obtained through DDE-BIFTOOL software, which is a collection of MATLAB programs. In this regard, the number of unstable eigenvalues is associated with specific line types, as illustrated in Fig. \ref{Fig:Bif1}. Here, one should note that the countable indefinite eigenvalues number can be gotten via the characteristic equation exponential terms. Meanwhile, only the eigenvalues which are close sufficient to the imaginary axis matter for stability. Therefore, we only consider the eigenvalues which satisfy Re$(\mu)>-3$.

    \noindent The ideas presented in this paper support the existence of more complex dynamics, e.g. periodic and non-periodic oscillations. Our analytical and numerical experiments have shown that delay term $\tau_r$ has an especially great impact on the model dynamics. Therefore, further analysis in terms of the biological and analytical meaning of $\tau_r$ should be performed. Besides, the backward bifurcation phenomenon in infectious dispersion formulations is often studied for the purpose of epidemic control. However, as seen in this paper, some of these models also support other types of bifurcations, including transcritical, period doubling, and Hopf bifurcations, under certain conditions. Therefore, it is essential to have insight into the role of these bifurcations from an epidemiological point of view. Further analysis on this topic is reserved for future papers. Dengue infection is an endemic health disease, commonly in the sub-tropical and tropical areas of the globe. Thus, it is significant to consider the possible re-occurrence of the outbreak by examining various strategies of optimal control and effective cost analysis.

\bibliographystyle{ieeetr}
\bibliography{dengue.bst}

\footnotetext{{$^*$}(Corresponding author)}
\end{document}